\documentclass{article}
\usepackage{arxiv}

\usepackage[utf8]{inputenc} 
\usepackage[T1]{fontenc}    
\usepackage{booktabs}       
\usepackage{amsfonts}       
\usepackage{nicefrac}       
\usepackage{microtype}      
\usepackage{lipsum}         
\usepackage{doi}

\usepackage{amsmath,amsfonts,mathrsfs,bm}
\usepackage{comment}

\usepackage[table]{xcolor}

\usepackage[pdftex]{graphicx}
\usepackage[title]{appendix}
\usepackage{amsthm,mathtools}
\usepackage{algorithm}
\usepackage{algpseudocode}
\usepackage{enumitem}
\usepackage{blindtext}
\usepackage{mdframed}
\usepackage{multicol}

\usepackage{dsfont}

\usepackage{subcaption}
\usepackage[justification=centering]{caption}

\usepackage{booktabs}

\usepackage{tikz}
\usetikzlibrary{arrows.meta}
\usetikzlibrary{positioning}
\usetikzlibrary{shapes,calc}
\usetikzlibrary{arrows,fit}

\usetikzlibrary{decorations.markings}

\usepackage{adjustbox}

\usepackage{array}
\newcolumntype{L}[1]{>{\raggedright\let\newline\\\arraybackslash\hspace{0pt}}m{#1}}
\newcolumntype{C}[1]{>{\centering\let\newline\\\arraybackslash\hspace{0pt}}m{#1}}
\newcolumntype{R}[1]{>{\raggedleft\let\newline\\\arraybackslash\hspace{0pt}}m{#1}}

\usepackage{multirow,multicol}

\usepackage{float}
\usepackage{pgffor}

\usepackage[section]{placeins}

\usepackage{chngcntr}
\counterwithin{figure}{section}
\counterwithin{table}{section}

\newtheorem{theorem}{Theorem}[section]

\newtheorem{prop}[theorem]{Proposition}
\newtheorem{corol}[theorem]{Corollary}
\newtheorem{rem}[theorem]{Remark}

\theoremstyle{definition}

\usepackage[sort,square,numbers]{natbib}

\usepackage{hyperref}
\usepackage{cleveref}

\crefalias{prop}{theorem}
\crefname{theorem}{proposition}{propositions}
\Crefname{theorem}{Proposition}{Propositions}

\newcommand{\KwIn}[1]{\Statex \textbf{Input:} #1}

\newcommand{\KwOut}[1]{\Statex \textbf{Output:}  #1}
\newcommand{\KwRet}[1]{\State \textbf{return}  #1}

\newcommand{\KwDepend}[1]{\Statex \textbf{Depends on:} #1}

\newenvironment{smallmat}{%
	\small\arraycolsep=0.8\arraycolsep
	\begin{pmatrix}}
	{\end{pmatrix}}

\definecolor{myred}{HTML}{B83232}
\definecolor{myblue}{HTML}{28568A}
\definecolor{mygreen}{HTML}{288A42}

\newcommand{\e}{\mathbb{E}}
\newcommand{\cov}{\mathrm{Cov}}
\newcommand{\var}{\mathrm{Var}}

\newcommand{\R}{\mathbb{R}}
\newcommand{\N}{\mathbb{N}}

\newcommand{\spn}[1]{\mathop{\mathrm{span}}\left\lbrace{#1}\right\rbrace}%
\newcommand{\ind}{\mathds{1}}
\newcommand{\q}[1]{``#1''}%
\newcommand{\dd}{\mathrm{d}}
\newcommand{\dv}{{\mathrm{div}}}
\newcommand{\curl}{\mathop{\mathrm{curl}}}

\newcommand{\lb}{-\Delta_{\mathcal{M}}}
\newcommand{\link}[2]{\href{#1}{\textcolor{myred}{#2}}}

\title{Prediction of spatio-temporal data on meshed surfaces using advection-diffusion SPDEs}

\author{ {\large Mike PEREIRA$^{1}$, Lucia CLAROTTO$^{2,\star}$ and Nicolas DESASSIS$^1$} \vspace{-1ex} \\ \vspace{-1ex}
 $^1$ STIM, Mines Paris -- PSL University, France.\\
 $^2$ UMR MIA PS, AgroParisTech/INRAE, Paris-Saclay University, France.\\
	$^{\star}$ Corresponding author : lucia.clarotto@agroparistech.fr
}

\renewcommand{\shorttitle}{\textsc{Prediction of spatio-temporal data on meshed surfaces using advection-diffusion SPDEs}}

\hypersetup{
	pdftitle={Prediction of spatio-temporal data on meshed surfaces using advection-diffusion SPDEs},
	pdfauthor={Mike Pereira, Lucia Clarotto, Nicolas Desassis}
}

\begin{document}
	
	
\maketitle
	

\begin{abstract}
    The aim of this work is to propose a statistical model for spatio-temporal data on meshed surfaces based on the Stochastic Partial Differential Equation (SPDE) modeling approach. Specifically, we focus on a class of advection-diffusion SPDEs defined on smooth compact orientable closed Riemannian manifolds of dimension 2, and their discretization via a Galerkin approach. Data following this model include measurements on tunnel surfaces, brain surfaces, or the globe. We demonstrate how this method enables the development of scalable algorithms for the simulation, inference and prediction of Gaussian random fields that are solutions to the discretized SPDE. The method is applied to a simulated spatio-temporal dataset exhibiting advective and diffusive behavior on the sphere, as well as to a real case study on aerosol optical depth in the atmosphere across the globe’s surface.
\end{abstract}


\let\thefootnote\relax\footnotetext{\paragraph{\textsc{Acknowledgment}} The authors would like to thank the \href{https://chaire-geolearning.org/}{Geolearning} chair for its financial support.}


\vspace{3em}

\begin{center}
\textit{
Note: This document contains links to animated figures colored in \textcolor{myred}{dark red}. All these animations are collected at: \textcolor{myred}{\href{https://mike-pereira.github.io/PRED_STRF/}{\url{https://mike-pereira.github.io/PRED_STRF/}}}.}
\end{center}


\vspace{2em}
\section{Introduction}

In Geostatistics, when modeling spatio-temporal data, the observed variable is seen as (a realization of) a space-time Gaussian random field (GRF), so that the mere characterization of its mean and covariance functions suffices to fully describe its statistical properties. In particular, these two functions are chosen to mimic the spatio-temporal variability and structure observed in the data \cite{wackernagel2003multivariate}. 
This probabilistic framework has many advantages. On the one hand, it allows to perform simulations of random fields with the same spatio-temporal structure as the one observed in the data, and predictions at unobserved locations. On the other hand, uncertainties can be quantified both on the variable behavior at unobserved locations (using so-called conditional simulations) or on the model parameters (through Bayesian approaches) \citep{chiles2009geostatistics}.

Of particular interest in this work is the setting where the spatial domain on which the data lie is not Euclidean, but rather represents a meshed surface. For instance, this is the case when dealing with fMRI data in neuroimaging applications, in which case the data lie on the cortical surface (i.e. the surface of the brain): see eg.~\cite{mejia2020bayesian}. This is also the case when considering global data in environmental applications, for which the data lie on a sphere representing our planet (see eg. \cite{rayner2020eustace}), and on this surface transport phenomena (due to winds and currents for instance) can affect the structure of the data (see eg. \Cref{fig:sulfate}). Hence the main motivations of this work: proposing models for spatio-temporal GRFs flexible enough to represent complex patterns of correlations in data lying on compact meshed surfaces without boundary, but simple enough so that numerically efficient algorithms for their simulation, inference and prediction can be derived.

\begin{figure}
	\centering
\centering
{\includegraphics[height=0.18\textheight]{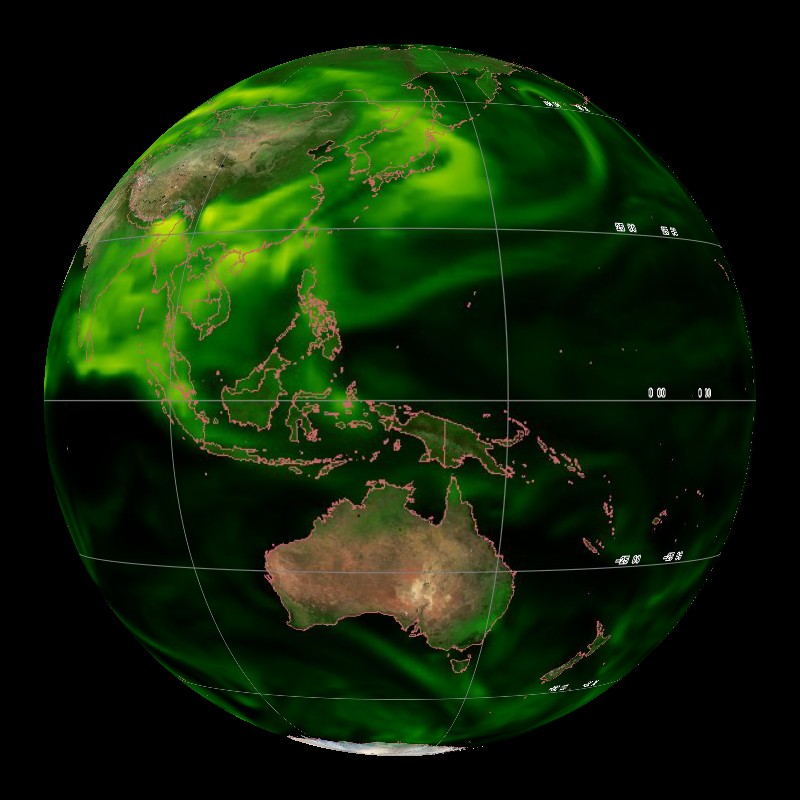}}
	\caption{Example of spatio-temporal data distributed on the sphere: \link{https://mike-pereira.github.io/PRED_STRF/sulfate}{Simulation of the presence of sulfate in the Earth atmosphere} (Source: \href{https://sos.noaa.gov/catalog/datasets/aerosols-sulfate/}{NASA Global Modeling and Assimilation Office})}\label{fig:sulfate}
\end{figure}%

\section{Context and state-of-the art}\label{sec:context}

When it comes to defining and building spatio-temporal GRFs to model data, there exists two main approaches: either through the definition of valid covariance functions (which are then \q{fitted} to the data), or through dynamical models describing the evolution in space and time of the GRFs. Let us review the principle and limits of both approaches.

\subsection{The covariance-based approach} Since we consider GRFs, tasks such as sampling from (un)conditional distributions, predictions (through conditional expectations) and likelihood-based inference can all be performed by solving linear systems or adequately factorizing covariance matrices of the field \cite{wackernagel2003multivariate}. Hence, the most straightforward (and classical) approach to spatio-temporal geostatistical modeling consists in specifying valid space–time covariance functions and fitting them to the data, where the parameters determine the scale and structure of spatial and temporal dependence, while ensuring positive definiteness of the associated covariance matrices. Consequently, extensive literature on which covariance functions may be used to model spatio-temporal data, even with complex correlation patterns, is available (see \cite{porcu202130,porcu2018modeling,chen2021space} for recent reviews).

Nonetheless, the covariance-based approach has two main drawbacks.
First, the matrix factorizations required in sampling, prediction and inference tasks have a complexity that scales as the cube of the number of observations, thus making them unfeasible when this last number is large. To circumvent this, simplifying assumptions on the covariance model must be made, such as the separability of space and time dependencies or stationarity. However, these restrictive assumptions fail to capture spatial heterogeneity, regime shifts and complex space–time dependence patterns. As a consequence, they may substantially reduce model realism, particularly for environmental processes where dynamics often vary across regions and evolve over time.
Secondly, since they rely on Euclidean or arc-length distances, most of the covariance models available in the literature are restricted to the setting where the spatial domain is either Euclidean or the sphere. Hence, they hardly generalize to other surfaces.

\subsection{The dynamic approach}

As foretold by its name, this approach relies on models of the dynamic evolution of the GRF in time and space. Such models offer an elegant framework for describing spatio-temporal evolution, especially when the underlying dynamics are complex. In addition, the model parameters are often interpretable and explicitly control the temporal evolution and spatial propagation mechanisms of the process.

Early approaches include vector autoregressive and moving-average formulations for spatially indexed time series, extending classical time-series ideas to the multivariate spatial setting \citep{cressie2011}. Integral–difference equation (IDE) models represent evolution via convolution operators that describe spatial propagation and redistribution over time, providing a flexible framework for large-scale environmental systems (e.g., \cite{wikle99, zammit2020}).

In this work, we specifically focus on dynamical models that take the form of stochastic partial differential equations (SPDE), the solutions of which are GRFs. This \q{SPDE approach} to GRF modeling has been popularized by \citet{lindgren2011explicit}, and builds on a result from \citet{whittle1954stationary} which states that isotropic GRFs $Z$ on $\R^d$ ($d\in\mathbb{N}$) with a Matérn covariance function are stationary solutions of the SPDE given by
\begin{equation}
	(\kappa^2-\Delta)^{\alpha/2} Z = \tau \mathcal{W},
	\label{eq:spde_space}
\end{equation}
where $\kappa >0$, $\alpha>d/2$, $\tau >0$, and $(\kappa^2-\Delta)^{\alpha/2}$ is a pseudo-differential operator (defined as $(\kappa^2-\Delta)^{\alpha/2}[\cdot]=\mathcal{F}^{-1}[\bm w\mapsto (\kappa^2+\Vert \bm w\Vert^2)^{\alpha/2}\mathcal{F}[\cdot](\bm w)]$), and $\mathcal{W}$ is a Gaussian white noise on $\R^d$. Solving numerically this SPDE using stochastic finite elements allows to directly obtain an expression for the precision matrix (i.e. the inverse of the covariance matrix) of a GRF with Matérn covariance. Then, at the price of a minor approximation (called mass lumping), these expressions yield a Gaussian Markov random field representation of the GRF, characterized by a sparse precision matrix \cite{rue2005gaussian,lindgren2022spde}. This in turn results in significant computational gains since sparse matrix algorithms can be used to deal with the matrix factorizations and linear system solving involved when performing sampling, prediction and inference \cite{lindgren2011explicit,davis2006direct}.

The SPDE approach has been extensively used to model spatial data on Euclidean domains (see \cite{lindgren2022spde} for a recent review), and extended to model spatial data on surfaces (see eg.~\cite{mejia2020bayesian,bonito2022numerical,coveney2019probabilistic}) and more generally on Riemannian manifolds (see eg. \cite{herrmann2020multilevel,lang2023galerkin}) by replacing the Laplace operator $-\Delta$ in SPDE~\eqref{eq:spde_space} by a Laplace--Beltrami operator. 

Extensions of the SPDE approach to the spatio-temporal setting have also been proposed. \citet{cameletti2013spatio} propose an approach where the spatial SPDE is coupled with an AR(1) process in time, thus yielding a separable model. Non-separable models based on a direct generalization of SPDE~\eqref{eq:spde_space} have been proposed by \citet{lindgren2024diffusion} and \citet{rayner2020eustace}, who consider the solutions of diffusion SPDE defined on Euclidean domains and on surfaces by
\begin{equation*}
	\frac{\partial Z}{\partial t}+(\kappa^2-\Delta)^{\alpha/2} Z = \tau \mathcal{W}_T \otimes \mathcal{Y}_S,
\end{equation*}
where $\mathcal{W}_T$ denotes a temporal white noise and $\mathcal{Y}_S$ denotes either a white or colored noise in space, whose detailed mathematical formulation is presented in \Cref{sec:model}. For reference, we provide in \Cref{fig:diff-sph} a simulation of a solution to this SPDE on the sphere. However, these models result in random fields with even covariance functions, meaning that changing the sign of the spatial or temporal lag at which the covariance is evaluated does not change the value of the covariance. Consequently, these models are incapable of accounting for transport effects such as advection phenomena (which are intrinsically asymmetrical in time). In the Euclidean setting, extensions of the SPDE approach allowing to deal with asymmetries in the covariance structure have been proposed by \citet{liu2022statistical,sigrist2015stochastic,clarotto2024spde}. Nonetheless, to the best of our knowledge, the generalization of such models to more complex geometries in spatio-temporal statistics is left open. We also note that extending SPDE models to a spatio-temporal framework is computationally challenging, as the classical use of Cholesky factorization for sparse matrices can suffer from significant fill-in, leading to increased computational and memory costs.

\begin{figure}
	\centering
	{\includegraphics[width=0.7\textwidth]{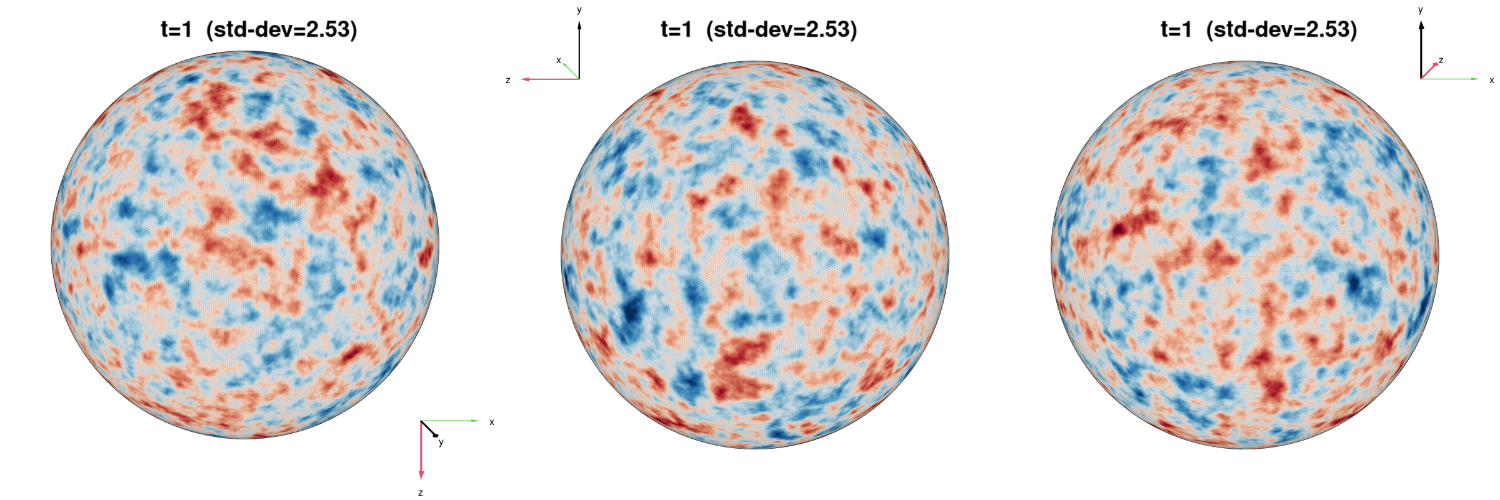}}\\
	\vspace{-1em}
	\includegraphics[width=0.2\textwidth]{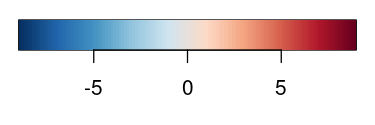}
	\caption{\link{https://mike-pereira.github.io/PRED_STRF/diff-sphere}{Simulation a spatio-temporal diffusion SPDE on the sphere represented from three different viewpoints on the surface (t represents the time step, std-dev the standard deviation of the field value across the surface).}}\label{fig:diff-sph}
\end{figure}

\subsection{Proposed approach}

The aim of this work is to propose new models for spatio-temporal data on meshed surfaces based on the SPDE modeling approach. To do so, we generalize the approach proposed by \citet{clarotto2024spde} to model spatio-temporal data using an advection-diffusion SPDE to compact smooth orientable Riemannian manifolds of dimension 2, without boundary.

Advection-diffusion models, unlike pure diffusion models, include a general vector field tangent to the Riemannian manifold, allowing them to capture directional transport and better explain environmental phenomena such as wind-driven or water-driven movement of particles.

Then, following the framework  proposed in \cite{lang2023galerkin}, we define a counterpart to the advection-diffusion SPDE on a meshed triangulation of the manifold using Galerkin approximations of the differential operators. We show how this approach allows to easily propose scalable algorithms for the simulation, inference and prediction of Gaussian random fields that are solutions to the resulting SPDE.

We demonstrate our methodology using aerosol optical depth (AOD) data on a high-resolution global grid. By modeling AOD as a spatio-temporal Gaussian process on a Riemannian manifold with wind-driven advection, our approach naturally incorporates the geometry of the Earth’s surface and the influence of atmospheric transport, enabling more accurate and physically informed predictions of dust concentration.

In this example, we consider advection to be driven by the observed, temporally varying wind vector field; however, alternative parameterizations of the advection field could also be used.

\section{Advection-diffusion SPDE on a Riemannian manifold}\label{sec:model}

We start by defining the advection-diffusion SPDE on a Riemannian manifold that will be considered in this work. Let $(\mathcal{M}, g)$ denote a compact orientable smooth Riemannian manifold of dimension $d=2$, without boundary. Let $T>0$, we denote by $\lb$ the Laplace--Beltrami operator of the surface.
We consider the following advection-diffusion SPDE on the domain $[0,T]\times \mathcal{M}$:
\begin{equation}
	\frac{\partial \mathcal{Z}}{\partial t} + \frac{1}{c}\big(P(\lb)\mathcal{Z} + \dv(\mathcal{Z}\gamma_t)\big)=\frac{1}{c}\mu_S+\frac{\tau}{\sqrt{c}}\mathcal{W}_T\otimes \mathcal{Y}_S,
	\label{eq:spde}
\end{equation}
where
\begin{itemize}
	\item $P$ is a polynomial such that for all $\lambda \ge 0$, $P(\lambda)\ge C_P$,  for some fixed $C_P >0$;
	\item $(t,s)\in[0,T]\times \mathcal{M}\mapsto \gamma_t(s)$ is a time-dependent smooth, bounded, \textit{divergence-free} vector field in $T\mathcal{M}$ (the so-called tangent bundle of $\mathcal{M}$, i.e. the collection of all tangent spaces of a manifold, organized smoothly into a single space), meaning that for any $t\in[0,T]$, $\gamma_t$ is a smooth vector field satisfying $\dv(\gamma_t)=0$, and the map $t\mapsto \gamma_t$ is smooth;
	\item $\mu_S \in L^2(\mathcal{M})$ is a deterministic, possibly null, time-invariant source term, capturing constant inputs that drive the system independently of advection and diffusion;
	\item $\mathcal{W}_T\otimes \mathcal{Y}_S$ is a space-time separable stochastic forcing defined as a Gaussian process whose covariance operator is the tensor product of the covariance of a time-dependent Gaussian white noise $\mathcal{W}_T$ and a space-dependent colored noise $\mathcal{Y}_S=f_S(-\Delta_{\mathcal{M}})\mathcal{W}_S$,
	where $f_S : \R_+ \rightarrow \R$ is a function for which there exists $\beta > d/4$ such that  $ \lambda \mapsto \lambda^{\beta}f_S(\lambda)$ is bounded on $\R_+$ (cf. \Cref{sec:funcL} for a proper definition of the operator $f_S(-\Delta_{\mathcal{M}})$  and of the colored noise);
    \item $c>0$ is a time-scaling parameter, and $\tau>0$ is a variance-scaling parameter;
    \item The initial condition $\mathcal{Z}(0,\cdot)$ is a $L^2(\mathcal{M})$-valued random variable.
\end{itemize} 

More precisely, the  forcing  term $\mathcal{W}_T\otimes \mathcal{Y}_S$ is defined as a generalized random field acting on functions of $L^2([0,T])\times L^2(\mathcal{M})$. Let $\langle\cdot,\cdot\rangle$ (resp.  $\langle\cdot,\cdot\rangle_T$) denote the usual inner product on $L^2(\mathcal{M})$ (resp. $L^2([0,T])$). Then for any $(\phi_T, \phi_S), (\varphi_T, \varphi_S)\in L^2([0,T])\times L^2(\mathcal{M})$, $\mathcal{W}_T\otimes \mathcal{Y}_S(\phi_T,\phi_S)$ and $\mathcal{W}_T\otimes \mathcal{Y}_S(\varphi_T,\varphi_S)$ are  centered Gaussian random variables, and 
\begin{equation*}
	\cov\bigg[\mathcal{W}_T\otimes \mathcal{Y}_S(\phi_T,\phi_S),\mathcal{W}_T\otimes \mathcal{Y}_S(\varphi_T,\varphi_S)\bigg]=\langle \phi_T, \varphi_T\rangle_T \langle f_S(-\Delta_{\mathcal{M}})\phi_S, f_S(-\Delta_{\mathcal{M}})\varphi_S\rangle.
\end{equation*}
Hence, SPDE~\eqref{eq:spde} can be rewritten in a perhaps more familiar form, for the readers used to stochastic differential equations (SDE) in infinite dimensions~\cite{da2014stochastic}, as 
\begin{equation}\label{eq:spde-sde}
	\dd \mathcal{Z} =-\frac{1}{c}\big(P(\lb)\mathcal{Z} + \dv(\mathcal{Z}\gamma_t)\big)\dd t 
	+\frac{1}{c}\mu_S+\frac{\tau}{\sqrt{c}}\dd \widetilde{\mathcal{W}}_t,
\end{equation}
where $\lbrace\widetilde{\mathcal{W}}_t\rbrace_{t\in [0,T]}$ denotes a cylindrical Wiener process (cf. Appendix~\ref{sec:analogy} for more details). Since the bilinear form $a_t(u,v) = \langle P(\lb)u + \dv(u\gamma_t), v\rangle$ is coercive (i.e., $\exists k\in\R, k'>0, \ a_t(u,u)\geq k'\Vert P(-\Delta)^{1/2}u \Vert_{L^2(\mathcal{M})}^2-k\Vert u \Vert_{L^2(\mathcal{M})}^2, \ \forall u, t$), and since the covariance operator of $\mathcal{Y}_S$ is trace-class (owing to the assumption on $f_S$), the existence and
uniqueness of a variational solution to SPDE~\eqref{eq:spde-sde} is ensured \citep[Chapter 2]{krylov1981stochastic}. In particular, the solution is continuous in time and has at least the spatial regularity induced by the norm $\Vert P(-\Delta)^{1/2}\cdot\Vert_{L^2(\mathcal{M})}$.

\begin{rem}
    Existence and uniqueness of the solution are also ensured when vector fields $\gamma_t$ with a non-zero divergence are considered. However, such vector fields should satisfy the condition
	\begin{equation}\label{eq:cond_gamma}
		\inf_{t\in [0,T]} C_P + \frac{1}{2}\dv(\gamma_t)\ge \eta_0
	\end{equation}
	for some fixed $\eta_0 >0$. This condition ensures the coercivity of the bilinear form $a_t(u,v)$ described above.
\end{rem}

\begin{rem}
	By linearity, a solution to~\eqref{eq:spde-sde} can be expressed as the sum $\mathcal{Z}=z_0+\mathcal{X}$, where  $z_0$ is a solution to the deterministic differential equation
	\begin{equation*}
		\dd z_0 =-\frac{1}{c}\big(P(\lb)z_0 + \dv(z_0\gamma_t)\big)\dd t 
		+\frac{1}{c}\mu_S,
	\end{equation*}
and $\mathcal{X}$ is a solution to the SDE
\begin{equation*}
	\dd \mathcal{X} =-\frac{1}{c}\big(P(\lb)\mathcal{X} + \dv(\mathcal{X}\gamma_t)\big)\dd t 
	+\frac{\tau}{\sqrt{c}}\dd \widetilde{\mathcal{W}}_t.
\end{equation*}
where in particular, $\e[\mathcal{Z}]=z_0$ and $\e[\mathcal{X}]=0$.
\end{rem}

Note that we can consider $P$ in \eqref{eq:spde} as a more general polynomial than  $(\kappa^2 - \Delta)$. This generalization allows us to extend our framework to operators associated with oscillating Whittle–Matérn fields (see, e.g., \cite[Section 3.3]{lindgren2011explicit}). It also enables the treatment of SPDEs linked to PDEs involving iterated Laplacians, which can appear as linearizations of non-linear operators (see eg. \cite{larsson2011finite}), or hyperviscosity terms, introduced for instance in fluid dynamics (see e.g. \cite{rot2026adaptive}).

Besides, we also restrict the model to divergence-free vector fields, which effectively means that we do not model, through the advection term, certain transport effects (eg. local accumulation or depletion due to the ``divergence'' part of the vector field). But our goal is not to reproduce the detailed dynamics exactly, but rather to inform the construction of a Gaussian process that captures the large-scale behavior and key statistical properties of the field. 
To that end, retaining only the divergence-free component of the wind is a good option since it will still preserve the main advection patterns. As for the discarded part of the wind field, along with other any unmodeled effects, we effectively treat them as a stochastic spatio-temporal source term.

We conclude this section with a few words about the vector fields $\gamma_t$, $t\in[0,T]$. A complete characterization of smooth vector fields on a manifold is given by the Helmholtz-Hodge decomposition theorem, which states that vector fields may be uniquely decomposed as the sum of a irrotational component (whose $\curl$ is zero), a divergence-free component (whose $\dv$ is zero) and a harmonic component \cite{bhatia2012helmholtz}. In particular, since we are considering manifolds of dimension $2$, this decomposition yields a natural parametrization for vector fields.
Starting from two scalar functions $\xi, \chi : {\mathcal{M}} \rightarrow \R$, we can obtain a tangent vector field $v$ as 
\begin{equation}
	v(s)=\nabla\xi(s) + \vec{n}(s)\times \nabla \chi(s) \in T_s\mathcal{M},
	\label{eq:pgamma2}
\end{equation}
with $\vec{n}(s)$ denoting the  vector normal to $\mathcal{M}$ and pointing outwards \cite{molina2018rapid}. In particular, $\nabla\xi$ then denotes the irrotational component of $v$ and $\vec{n}\times \nabla \chi$ its divergence-free component in the Helmholtz-Hodge decomposition \cite{bhatia2012helmholtz}. 
Hence, a divergence-free vector field would only require the scalar potential function $\chi$ to be specified/parametrized. Considering smooth time-dependent scalar functions $\xi$ and $\chi$ then allows to tackle the parametrization of vector fields such time-dependent vector fields $\gamma_t$ through the easier angle of parameterizing scalar functions on the manifold $\mathcal{M}$.

\section{Advection-diffusion SPDE on a meshed surface}\label{sec:theory}

	\subsection{Definition and discretization of the SPDE}

Let  $\mathcal{M}_h$ be a discretization of the manifold $\mathcal{M}$ into a polyhedral surface with mesh size $h>0$ (by triangulation). In particular, following the surface finite element approach, we assume that the nodes of the polyhedral surface $\mathcal{M}_h$ lie on the surface $\mathcal{M}$, and that $\mathcal{M}_h$ is close enough to $\mathcal{M}$ that there exists a smooth and invertible function $\ell : \mathcal{M}_h \rightarrow \mathcal{M}$ that maps any point of $\mathcal{M}_h$ to a unique point of $\mathcal{M}$ and vice-versa (i.e. $\mathcal{M}_h$ can be \q{lifted} to $\mathcal{M}$). An illustration of such an approximation for the sphere is provided in \Cref{fig:mesh_sphere}

\begin{figure}
    \centering
    \includegraphics[width=0.3\linewidth,trim={5cm 5cm 5cm 5cm},clip]{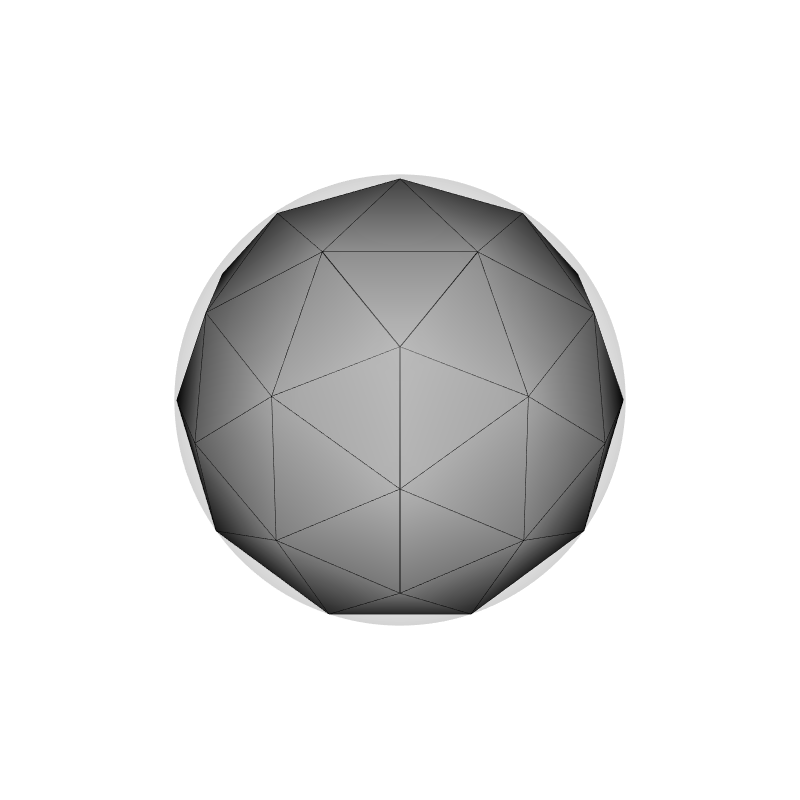}
    \caption{Illustration of a (coarse) polyhedral approximation $\mathcal{M}_h$ (in dark gray) when the manifold $\mathcal{M}$ is the sphere (light gray).}
    \label{fig:mesh_sphere}
\end{figure}

Let $ \lbrace \psi_1, \dots, \psi_N\rbrace \subset H^1(\mathcal{M}_h)$ be the \textit{linear} finite element basis associated with $\mathcal{M}_h$, where $N$ is the number of nodes of the triangulation.  $H^1(\mathcal{M}_h)$ denotes a Sobolev space, i.e., roughly speaking, the space of smooth enough functions on the surface $\mathcal{M}_h$, with finite values and gradients. Let then $V_N=\spn{\psi_k : 1\le k\le N}$ and $V_N^\ell =\spn{\psi_k^\ell : 1\le k\le N} \subset H^1(\mathcal{M})$ be the set of lifted finite element functions : for any $1\le k\le N$,  $\psi_k^\ell = \psi_k\circ\ell^{-1}$. 

In order to formulate a discretization of the advection-diffusion SPDE~\eqref{eq:spde}, we
look for an approximation of the solution $\mathcal{Z}$ of~\eqref{eq:spde} that can be expressed as a $V_N^\ell$-valued random variable (at any time). We obtain it by approximating each term in the SPDE by operators or variables that \q{live} on the lifted finite element space $V_N^\ell$, using so-called Galerkin approximations (cf. Appendix~\ref{sec:Galerkin}). 

Let $\dv_N(\gamma_t\,\cdot)$  and $-\Delta_{N}$ be the Galerkin approximations of, respectively, $\dv(\gamma_t\,\cdot)$ and $\lb$ over $V_N^\ell$ (cf. Appendix~\ref{sec:Galerkin}). Let $\lbrace \lambda_k^{(N)}\rbrace_{1\le k\le N}$ denote the eigenvalues of $-\Delta_N$, and  $\lbrace e_k^{(N)}\rbrace_{1\le k\le N}$ be a set of associated eigenfunctions forming an orthonormal basis of $V_N^\ell$.
First, the colored noise $\mathcal{Y}_S$ is approximated by a $V_N^\ell$-valued random variable $Y_S$ defined by
\begin{equation}\label{eq:noise_sd}
	{Y}_S = f_S(-\Delta_N)W_S=\sum_{k=1}^{N}w_kf_S(\lambda_k^{(N)}) e_k^{(N)},
\end{equation}
where $\lbrace w_k\rbrace_{1\le k\le N}$ a sequence of independent standard Gaussian variables and $f_S$ is defined as in the previous section. Note that the definition of $Y_S$ is equivalent to the definition of the colored noise $\mathcal{Y}_S$, after replacing $H$ by $V_N^\ell$ and the Laplace--Beltrami operator $\lb$ by its Galerkin approximation $-\Delta_N$. 

 Then, we define an approximation ${Z}^\ell(t,\cdot)\in V_N^\ell$ of the field $\mathcal{Z}(t,\cdot)$ by considering  SPDE~\eqref{eq:spde} and replacing the operators $\lb$ and $\dv(\gamma_t\,\cdot)$ by their Galerkin approximations, thus giving
\begin{equation}
	\frac{\partial{{Z}^\ell}}{\partial t} + \frac{1}{c}\bigg(P(-\Delta_N){{Z}^\ell} + \dv_N(\gamma_t {Z}^\ell)\bigg)=\frac{1}{c}M_S+\frac{\tau}{\sqrt{c}}\mathcal{W}_T\otimes {Y}_S, \quad t\in [0,T],
	\label{eq:spde_discr}
\end{equation}
where $M_S$ denotes the $L^2$-orthogonal projection of $\mu_S$ onto $V_N^\ell$, and $\mathcal{W}_T\otimes {Y}_S$ is defined in the same way as its counterpart $\mathcal{W}_T\otimes \mathcal{Y}_S$, i.e. as a generalized random field acting on functions of $L^2([0,T])\times V_N^\ell$ (which can be identified with a cylindrical Wiener process, cf. Appendix~\ref{sec:analogy}).

We now discretize~\eqref{eq:spde_discr} in time, by applying an implicit Euler scheme with time step $\delta t>0$.  We start by discretizing the time interval $[0,T]$  into $K+1$ regular time steps of size $\delta t=T/K$, and write $t_k=k\delta t$ for $k\in \lbrace 0, \dots, K\rbrace$.
Let then ${Z}^{(k)}={Z}^\ell(t_k,\cdot)$ denote the approximation of the spatial trace of the solution to SPDE~\eqref{eq:spde_discr} at time $t_k$, and take $\gamma^{(k)}=\gamma_{t_k}$. We start from an initial condition ${Z}^{(0)}={Z}^\ell(0,\cdot)$ taken as a $V_N^\ell$-valued random variable, for instance the projection of the (random) initial condition $\mathcal{Z}(0,\cdot)$ of the original SPDE~\eqref{eq:spde} onto $V_N^\ell$. 
Then, we have the recursion 
\begin{equation}
	{Z}^{(k+1)}-{Z}^{(k)}
	+\frac{\delta t}{c}\bigg(P(-\Delta_N){Z}^{(k+1)} + \dv_N(\gamma^{(k)} Z^{(k+1)})\bigg)
	=\frac{\delta t}{c}M_S+\tau\sqrt{\frac{\delta t}{c}}{Y}^{(k+1)}, \quad 0\le k<K,
	\label{eq:euler_discr}
\end{equation}
where  $\lbrace Y^{(k)}\rbrace_{k\in\N}$ is a sequence of independent samples of $Y_S$.

To facilitate computations and simulations, the following proposition provides an explicit matrix-based formulation of the Euler recursion. First, let us introduce some standard finite element discretization matrices. We denote by $\bm C$, $\bm R$ and $\bm B^{(k)}$ ($0\le k< K$) the matrices whose entries are respectively given by
\begin{equation}\label{eq:coef}
	C_{ij}=\langle \psi_i^\ell, \psi_j^\ell\rangle, 
	\quad R_{ij}=\langle \nabla\psi_i^\ell, \nabla\psi_j^\ell\rangle,
	\quad B_{ij}^{(k)}=\langle \psi_i^\ell, \dv(\gamma^{(k)} \psi_j^\ell)\rangle, \quad 1\le i,j\le N,
\end{equation}
and let $\sqrt{\bm C}\in\R^{N\times N}$ such that $\bm C = \sqrt{\bm C}(\sqrt{\bm C})^{T}$. We also introduce the scaled matrices $\tilde{\bm R}$ and $\tilde{\bm B}^{(k)}$ defined by
\begin{equation*}
	\widetilde{\bm R} = (\sqrt{\bm C})^{-1}\bm R(\sqrt{\bm C})^{-T}, \quad 
	\widetilde{\bm B}^{(k)}=(\sqrt{\bm C})^{-1}\bm B^{(k)}(\sqrt{\bm C})^{-T}.
\end{equation*}

\begin{prop}\label{prop:rec_ad}
	Let $\bm \mu^{(0)}=(\mu_1^{(0)},\dots,\mu_N^{(0)})$, $\bm m=(m_1,\dots,m_N)$ and, for $0\le k\le K$, let $\bm z^{(k)}=(z_1^{(k)},\dots, z_N^{(k)})^T$ such that the time-discretization  $\lbrace {Z}^{(k)}\rbrace_{0\le k\le K}$ of SPDE~\eqref{eq:spde_discr} satisfy
	\begin{equation}
		M_S=\sum_{j=1}^N m_j\psi_j^\ell, \quad
		\e[Z^{(0)}]=\sum_{j=1}^N \mu_j^{(0)}\psi_j^\ell, \quad\text{and}\quad
		Z^{(k)}=\sum_{j=1}^N z_j^{(k)}\psi_j^\ell.
		\label{eq:dd}
	\end{equation}

	We also denote, for $0\le k<K$, by ${\bm \Gamma}^{(k)}$ the matrix defined by
	\begin{equation}
		{\bm \Gamma^{(k)}}=\bm I  +\frac{\delta t}{c} \big(P(\widetilde{\bm R}) + \widetilde{\bm B}^{(k)}\big).
		\label{eq:Gamma}
	\end{equation}
	Let $\bm x^{(0)}=(\sqrt{\bm C})^T \bm z^{(0)}$. Then we have the following recursion for $0\le k < K$:
	\begin{equation}
	\left\lbrace\begin{aligned}
			& \bm\Gamma^{(k)}\bm x^{(k+1)}
			=\bm x^{(k)}+\frac{\delta t}{c}(\sqrt{\bm C})^T \bm m+{f}_{\delta t}(\widetilde{\bm R}) \bm w ^{(k+1)},\\
			&\bm z^{(k+1)}=(\sqrt{\bm C})^{-T} \bm x^{(k+1)},
		\end{aligned}\right.
		\label{eq:rec_ad}
	\end{equation}
	where $\lbrace\bm w ^{(k)}\rbrace_{k\in\N}$ is a sequence of independent centered Gaussian vectors with covariance matrix $\bm I$, and ${f}_{\delta t}$ is the function defined by
	\begin{equation*}
		{f}_{\delta t}(\lambda)=\tau\sqrt{\frac{\delta t}{c}}f_S(\lambda), \quad \lambda\ge 0,
	\end{equation*}
and $f_{\delta t}(\widetilde{\bm R})$ is the associated matrix function (cf. Appendix~\ref{sec:funcMat} for a definition of matrix functions). 
\end{prop}

The proof is available in Appendix~\ref{proof:rec_ad}.

In practice, the approximate solution $Z^\ell$ in \eqref{eq:spde_discr} is computed directly on the polyhedral surface $\mathcal{M}_h$, and lifted to the original surface afterwards. To do so, we simply consider the recursion in \Cref{prop:rec_ad}, but replace the finite element matrices $\bm C$, $\bm R$ and $\bm B^{(k)}$ by their approximations on $\mathcal{M}_h$. These approximations are obtained by replacing the surface integrals on $\mathcal{M}$  by integrals computed on its polyhedral approximation $\mathcal{M}_h$ when computing the entries of $\bm C$, $\bm R$ and $\bm B^{(k)}$, thus giving
\begin{equation}\label{eq:coef-approx}
	C_{ij}=\langle \psi_i, \psi_j\rangle_{\mathcal{M}_h}, 
	\quad R_{ij}=\langle \nabla\psi_i, \nabla\psi_j\rangle_{\mathcal{M}_h},
	\quad B_{ij}^{(k)}=\langle \psi_i, \dv(\gamma^{(k)} \psi_j)\rangle_{\mathcal{M}_h}, \quad 1\le i,j\le N,
\end{equation}
where $\langle\cdot,\cdot\rangle_{\mathcal{M}_h}$ denotes the inner product over $L^2({\mathcal{M}_h})$, and the gradient and $\dv$ operator are defined along each (planar) face of $\mathcal{M}_h$. 
Such integrals can be easily computed on each face of $\mathcal{M}_h$ and then summed over the entire mesh. The resulting field $Z$, which lives in $V_N\subset H^1(\mathcal{M}_h)$ is then lifted to $\mathcal{M}$ to obtain a field $Z^\ell$ in $V_N^\ell\subset H^1(\mathcal{M})$. Details and convergence results linked to this approach can be found in \cite{lang2023galerkin}. In the remainder of the text, and unless specified otherwise, we assume that entries are computed according to this approach.

Note that within this approach, the vector field $\gamma_t$ is approximated by a vector field on the polyhedral surface $\mathcal{M}_h$ by identifying each triangle $T$ composing $\mathcal{M}_h$ with the tangent space of $\mathcal{M}$ at the lifted barycenter $s_T\in \mathcal{M}$ of $T$. Then, the vector field $\gamma_t$ is replaced by the vector field of $\mathcal{M}_h$, constant on each triangle $T$ forming $\mathcal{M}_h$, and taking the values $\gamma_t(s_T)$ on $T$. For ease of notations, we also use the notation $\gamma_t$ to refer to such an approximation of the vector field on the polyhedral surface $\mathcal{M}_h$.

\begin{rem}
  When the decomposition \eqref{eq:pgamma2} is used to parametrize the vector field $\gamma_t$ of $\mathcal{M}$ (using time-dependent functions $\xi$ and $\chi$), the corresponding vector field on $\mathcal{M}_h$ can be obtained by considering the potential function $\chi$ (and $\xi$ if vector fields with non-zero divergence were considered) in $V_N^\ell$, which yields a piecewise constant approximation of the vector field. Details about the computations related to such vector fields can be found in \cite{poelke2016boundary}.  
\end{rem}

We present in \Cref{fig:sim_ad} an example of simulation to the  advection-diffusion SPDE on a meshed sphere. We take  $M_S=0$ (which implies that $\bm m=\bm 0$), $\bm \mu^{(0)}=\bm 0$, $P(\lambda)=\kappa^2+\lambda$, $f_S(\lambda)=(\kappa_S^2+\lambda)^{-1}$ and $f_0(\lambda)=(\kappa_{in}^2+\lambda)^{-1}$. A time-invariant advection vector field was parametrized thanks to two scalar functions as in~\eqref{eq:pgamma2}. As seen in the simulation, the resulting random fields can reproduce complex advection and diffusion phenomena. 

\begin{figure}
\centering
\begin{subfigure}[c]{\textwidth}
\centering
\includegraphics[width=0.7\textwidth]{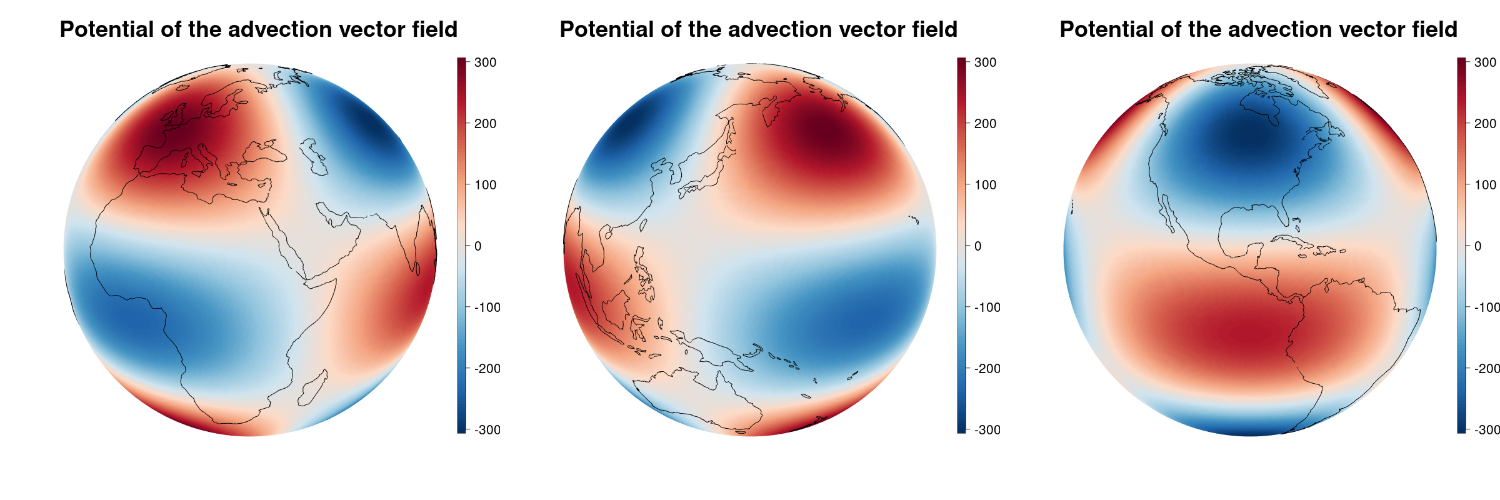}
\caption{Advection potential function $\chi$ (divergence-free component).}
\end{subfigure}
\begin{subfigure}[c]{\textwidth}
	\centering
{\includegraphics[width=0.7\textwidth]{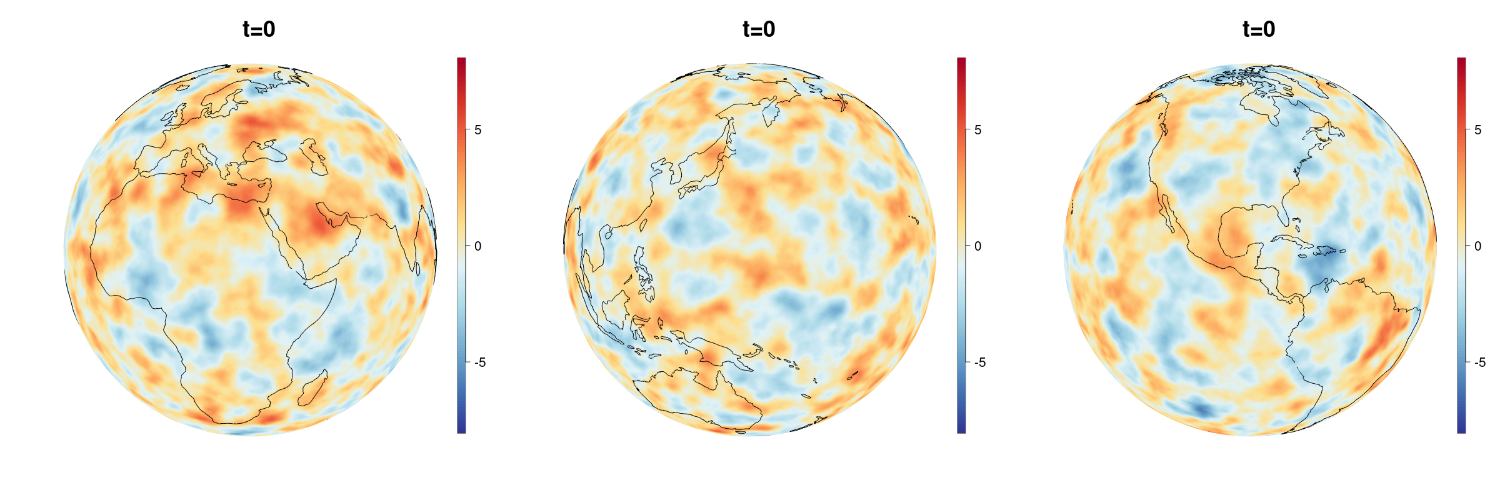}}
	\caption{Numerical solution of the SPDE at $t=0$.}
\end{subfigure}
\begin{subfigure}[c]{\textwidth}
	\centering
{\includegraphics[width=0.7\textwidth]{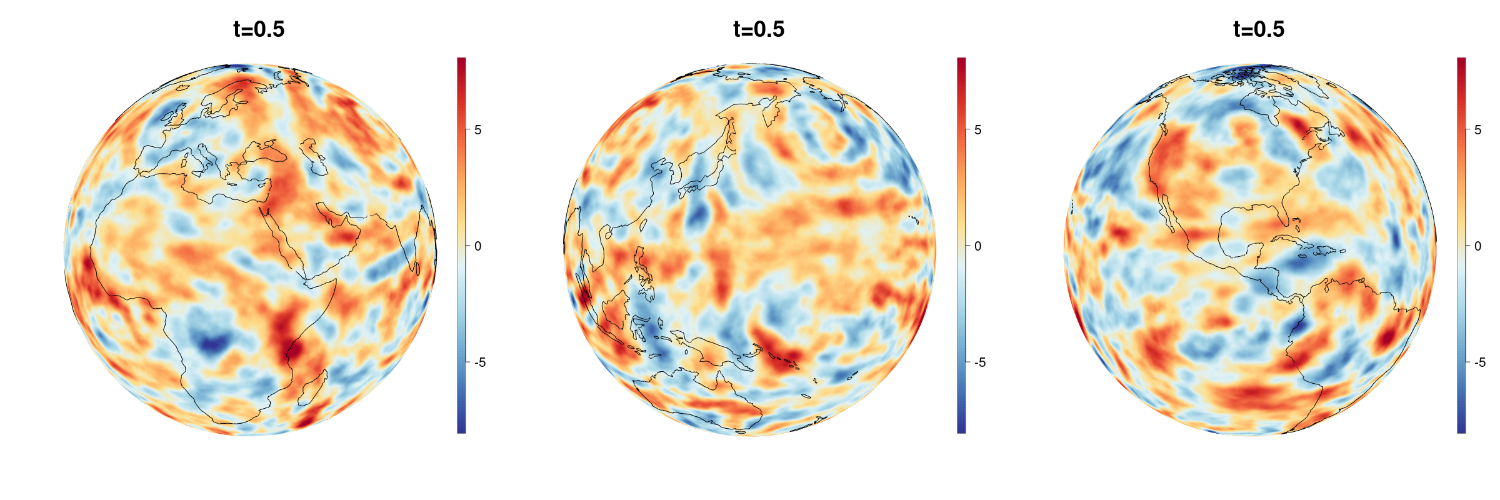}}
	\caption{Numerical solution of the SPDE at $t=0.5$.}
\end{subfigure}
\begin{subfigure}[c]{\textwidth}
	\centering
	{\includegraphics[width=0.7\textwidth]{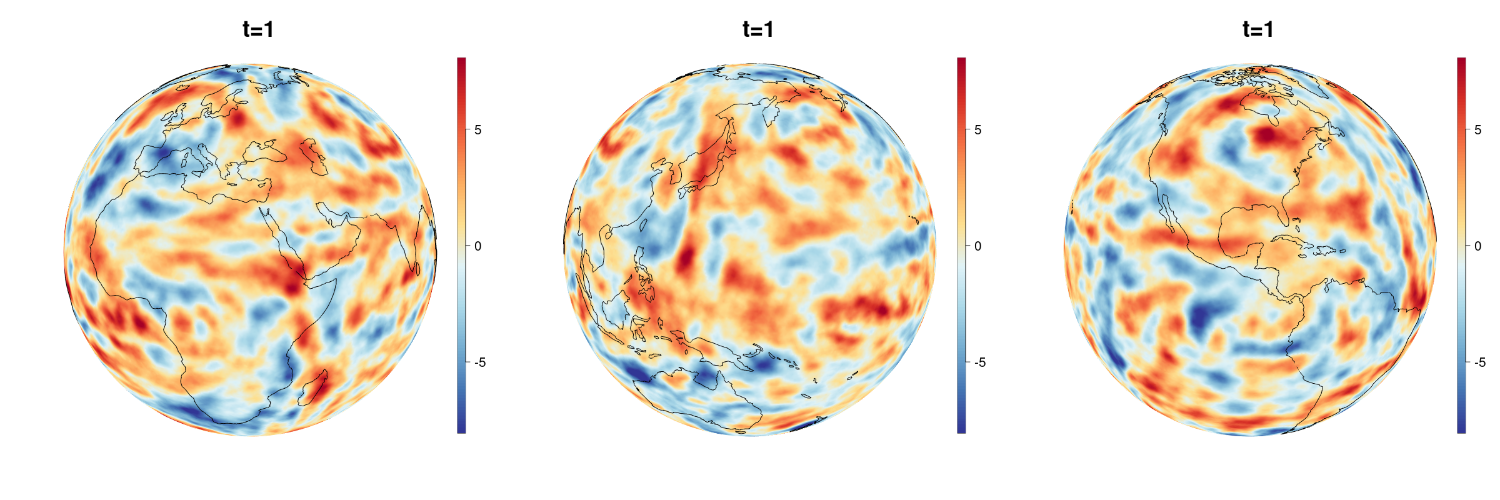}}
	\caption{Numerical solution of the SPDE at $t=1$.}
\end{subfigure}
\caption{\link{https://mike-pereira.github.io/PRED_STRF/adv-diff-sphere}{Simulation of a solution to the advection-diffusion SPDE on a meshed sphere, represented from three different viewpoints on the surface.}}\label{fig:sim_ad}
\end{figure}

Having established a discrete-time representation of the SPDE, we are now in a position to characterize the mean and covariance of the field at each time step.

\subsection{Mean and precision of the field}

The recursion in~\Cref{prop:rec_ad} can be used to derive, for any $K\in\N$, the expression of the mean and of the precision matrix of the vectors $\bm z^{(0)}, \dots, \bm z^{(K)}$. 
 Let us introduce some notations. For $\bm \Theta^{(0)},\dots, \bm \Theta^{(K-1)},\bm \Theta_1,\bm \Theta_2\in\R^{N\times N}$, let $\bm L(\lbrace\bm \Theta^{(k)}\rbrace_{k=0}^{K-1})$, $\bm D(\bm \Theta_1,\bm \Theta_2)$, $\bm D(\bm \Theta_1)\in~\R^{(K+1)N\times (K+1) N}$ be the block matrices  given by
\begin{equation}
	\bm L(\lbrace\bm \Theta^{(k)}\rbrace_{k=0}^{K-1})
	=\begin{smallmat}
		\bm I & &  &   \\
		-\bm I & {\bm \Theta^{(0)}} & &  \\
		& \ddots & \ddots  &   \\
		& &   -\bm I & {\bm \Theta^{(K-1)}}\\
	\end{smallmat},  
	\bm D(\bm \Theta_1,\bm \Theta_2)
	=\begin{smallmat}
		{\bm \Theta_1} & &  &   \\
		& {\bm \Theta_2} & &  \\
		&  & \ddots  &   \\
		& &    & {\bm \Theta_2}\\
	\end{smallmat},  \text{and} \quad 
	\bm D(\bm \Theta_1)=\bm D(\bm \Theta_1,\bm \Theta_1).
	\label{eq:defL}
\end{equation}

The next proposition gives the expression of the precision matrix of the coefficients $\bm z^{(0)}, \dots, \bm z^{(K)}$ obtained through the recursion~\eqref{eq:rec_ad}.

\begin{prop}\label{prop:Qeuler}
	Let us assume that the initial condition of SPDE~\eqref{eq:spde_discr} can be expressed as
	\begin{equation}\label{eq:ic}
		Z^{(0)}=\mu^{(0)}+f_0(-\Delta_N)W_S
	\end{equation}
	for some (deterministic) $\mu^{(0)}\in V_N$  and some  $f_0 : \R_+\rightarrow \R$ that is bounded takes positive values (meaning in particular that $\e[Z^{(0)}]=\mu^{(0)}$).
	Let then $\bm Z,{\bm M}_{\delta t}\in\R^{(K+1)N}$ be defined by
	\begin{equation*}
		\bm Z= \begin{smallmat}
			\bm z^{(0)}   \\
			\bm z^{(1)}   \\
			\vdots   \\
			\bm z^{(K)}\\
		\end{smallmat}
		\quad \text{and}\quad
		{\bm M}_{\delta t}= \begin{smallmat}
			\bm \mu^{(0)}   \\
			\frac{\delta t}{c}{\bm m}   \\
			\vdots   \\
			\frac{\delta t}{c}{\bm m}\\
		\end{smallmat}
	\end{equation*}
	where the vectors the vectors  ${\bm \mu}^{(0)}$, ${\bm m}$, and $\lbrace\bm z^{(k)}\rbrace_{0\le k\le K}$ are given in \Cref{prop:rec_ad}. Let $\bm \mu_{\bm Z}$ be the expectation and $\bm Q_{\bm Z}$ be the precision matrix of $\bm Z$.
	Then, we have 
	\begin{equation}
		\bm \mu_{\bm Z}=\e[\bm Z]=\bm D\big((\sqrt{\bm C})^{-T})\bm L\left(\lbrace\bm \Gamma^{(k)}\rbrace_{k=0}^{K-1}\right)^{-1}\bm D\big((\sqrt{\bm C})^{T}\big){\bm M}_{\delta t},
		\label{eq:defExp}
	\end{equation}
	and
	\begin{equation}
		\bm Q_{\bm Z}=\e[(\bm Z-\bm \mu_{\bm Z})(\bm Z-\bm \mu_{\bm Z})^T]^{-1} =  \bm D\big(\sqrt{\bm C}\big)~\bm L\left(\lbrace\bm \Gamma^{(k)}\rbrace_{k=0}^{K-1}\right)^T~\bm D\left(f_0^{-2}(\widetilde{\bm R}),\, {f}_{\delta t}^{-2}(\widetilde{\bm R})\right)~\bm L\left(\lbrace\bm \Gamma^{(k)}\rbrace_{k=0}^{K-1}\right)~\bm D\big((\sqrt{\bm C})^{T}\big).
		\label{eq:defQ}
	\end{equation}
	
\end{prop}

The proof is available in Appendix~\ref{proof:Qeuler}.

The explicit computation of the precision matrix  $\bm Q_{\bm Z}$ results in a block tri-diagonal matrix given by
\begin{equation*}
    \bm Q_{\bm Z}=	\begin{smallmat}
        \bm\Gamma_1^{(0)} & -\bm\Gamma_2^{(0)} &  & \\
        -(\bm\Gamma_2^{(0)})^T& \bm\Gamma_1^{(1)} & -\bm\Gamma_2^{(1)}\\
        & &\\
        & \ddots & \ddots  & \ddots & \\
        & &  - (\bm\Gamma_2^{(K-2)})^T &  \bm\Gamma_1^{(K-1)} & -\bm\Gamma_2^{(K-1)} \\
        & & &  -(\bm\Gamma_2^{(K-1)})^T & \bm\Gamma_1^{(K)}\\
    \end{smallmat},
\end{equation*}
where 
\begin{equation*}
    \bm\Gamma_1^{(k)}
    =
    \begin{cases}
    (\sqrt{\bm C})\big(f_0^{-2}(\widetilde{\bm R}) + {f}_{\delta t}^{-2}(\widetilde{\bm R})\big)(\sqrt{\bm C})^{T} & \text{if } k=0 \\
    (\sqrt{\bm C})\left({(\bm \Gamma^{(k-1)})}^{T}{f}_{\delta t}^{-2}(\widetilde{\bm R})~{\bm \Gamma^{(k-1)}} + {f}_{\delta t}^{-2}(\widetilde{\bm R})\right)(\sqrt{\bm C})^{T} & \text{if } 1\le k\le K-1 \\
    (\sqrt{\bm C}){(\bm \Gamma^{(K-1)}})^{T}{f}_{\delta t}^{-2}(\widetilde{\bm R})~{\bm \Gamma^{(K-1)}}(\sqrt{\bm C})^{T} & \text{if } k=K \\
    \end{cases}
\end{equation*}
and $\displaystyle \bm\Gamma_2^{(k)}=(\sqrt{\bm C}){f}_{\delta t}^{-2}(\widetilde{\bm R}){\bm \Gamma^{(k)}}(\sqrt{\bm C})^{T}$ for $0\le k\le K-1$.

As an illustration, starting from the same model as the one simulated in \Cref{fig:sim_ad}, we represent in \Cref{fig:cov} the spatio-temporal evolution of the covariance between the value of the field at time $t=0$ at a reference point (in orange), and the values of the field elsewhere and at later times. These covariances are computed using the formula~\eqref{eq:defQ} in \Cref{prop:Qeuler} for the spatio-temporal precision matrix of the field. As one can note in the animation, the zone of high-correlation \q{moves} along the advection direction, as expected in an advection problem.

\begin{figure}
	\centering
	\begin{subfigure}[c]{0.33\textwidth}
	\centering
	{\includegraphics[width=\textwidth]{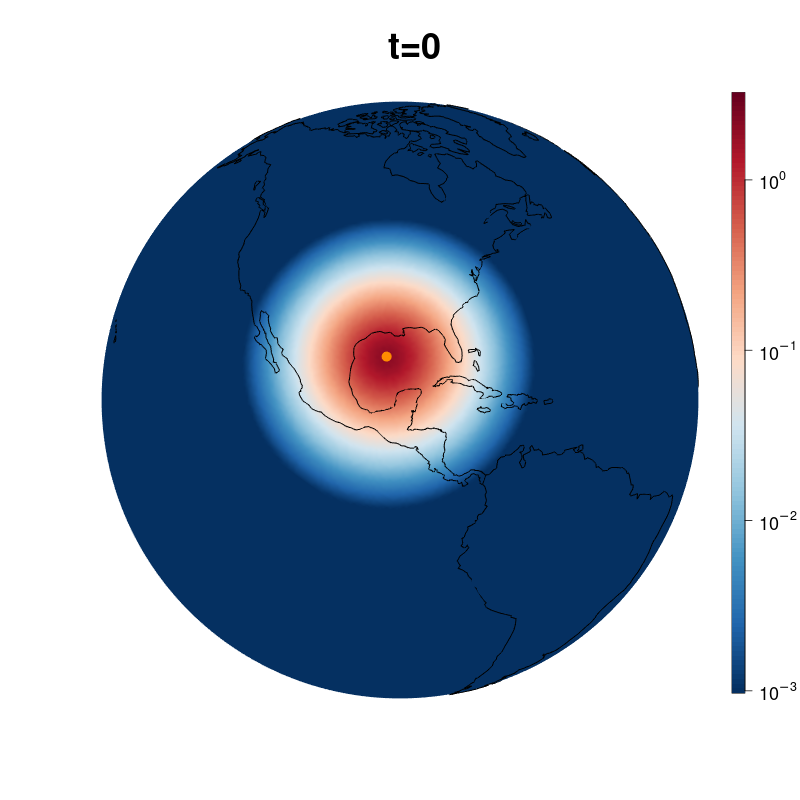}}
	\caption{Covariance at $t=0$.}
\end{subfigure}
\begin{subfigure}[c]{0.33\textwidth}
	\centering
	{\includegraphics[width=\textwidth]{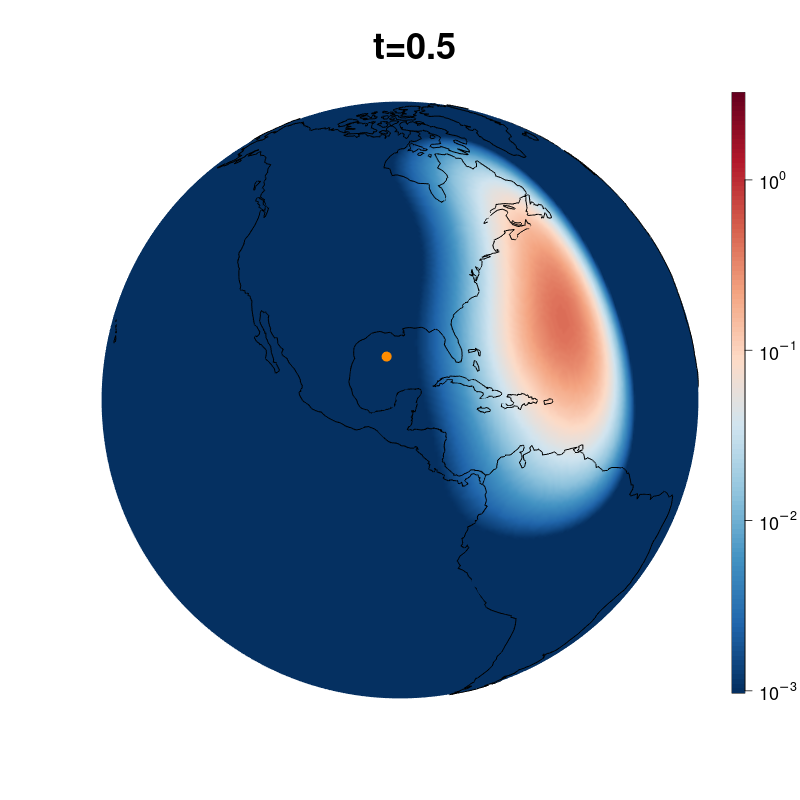}}
	\caption{Covariance at $t=0.5$.}
\end{subfigure}
\begin{subfigure}[c]{0.33\textwidth}
	\centering
	{\includegraphics[width=\textwidth]{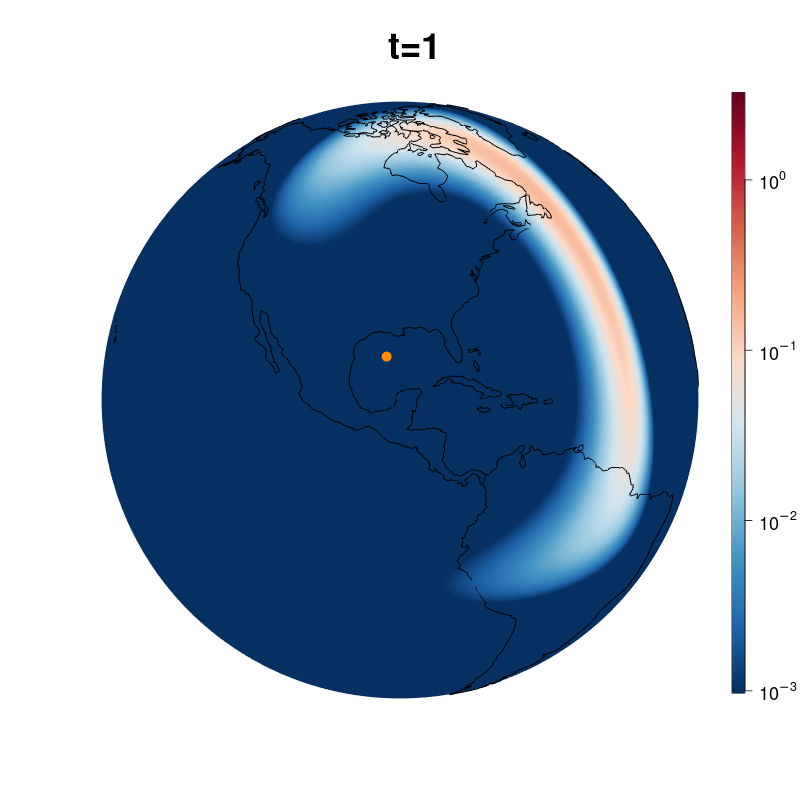}}
	\caption{Covariance at $t=1$.}
\end{subfigure}
	\caption{\link{https://mike-pereira.github.io/PRED_STRF/adv-diff-sphere-cov}{Spatio-temporal evolution of the covariance between a reference point (in orange) and the rest of the points in the domain, for the spatio-temporal model simulated in \Cref{fig:sim_ad}. The color (red to blue) represent the value of the covariance in log-scale.}}\label{fig:cov}
\end{figure}

The covariance matrix $\bm Q_{\bm Z}$ of the Euler discretization of $Z$ can be expressed as the product of block diagonal and block bi-diagonal matrices. This particular structure allows proposing scalable algorithms to perform products with vectors, solving linear systems and computing the log-determinant, all of which will prove useful when tackling inference and prediction problems (cf. \Cref{sec:inf_pred}). 

Having defined the recursion and characterized its mean and covariance, we next examine conditions under which the recursion remains stable over time.

\subsection{Stability of the recursion}

Similarly to the Euclidean case, the Galerkin approximation of advection-diffusion SPDEs on surfaces is subject to instabilities when the advection term dominates the diffusion term \cite{mekuria2016adaptive}. These instabilities are due to the fact that in advection-dominated cases, the matrices $\lbrace\bm \Gamma^{(k)}\rbrace_{k\in\N_0}$ in \eqref{eq:Gamma} dominate and can cause the linear systems involving $\bm\Gamma^{(k)}$ to become ill-conditioned.

Given these potential instabilities, a natural question is whether the recursion in~\Cref{prop:rec_ad} is stable. By stability, we mean here that the mean and covariance matrix of the vectors $\bm z^{(k)}$ stay at least bounded for any $k\in\N_0$. To ease the exposition, we discuss this property for the particular case where the parameters of the SPDE~\eqref{eq:spde} are time-invariant. Hence, we consider that the advection vector field $\gamma_t$ is constant through time, which in turn allows to introduce matrices $\widetilde{\bm B}$ and ${\bm \Gamma}$ such that for any $k$,   $\widetilde{\bm B}=\widetilde{\bm B}^{(k)}$ and  ${\bm \Gamma}={\bm \Gamma}^{(k)}$. Then, it turns out that under a mild condition on the (symmetric part of) the advection matrices $\widetilde{\bm B}$ we can ensure that the sequences $\lbrace \bm\mu^{(k)} \rbrace_{k\in\N_0}$ and $\lbrace \bm\Sigma^{(k)} \rbrace_{k\in\N_0}$ are not only bounded, but also convergent. To do so, we start by giving explicit expressions for the mean and covariance of the field.

\begin{prop}\label{prop:mean_cov_ad}
	We consider the notations introduced in \Cref{prop:rec_ad}, and assume that for  $k\ge 0$, ${\bm \Gamma}^{(k)}={\bm \Gamma}$ for some fixed matrix ${\bm \Gamma}$. Let then $\bm\mu^{(k)} = \e[\bm z^{(k)}]$ and $\bm\Sigma^{(k)}=\var[\bm z^{k}]=\e[(\bm z^{k}-\bm\mu^{(k)})(\bm z^{k}-\bm\mu^{(k)})^T]$ be respectively the mean and the covariance matrix of the vector $\bm z^{(k)}$. Then, for any $k\ge 0$,
    \begin{equation}\label{eq:exp_z}
		\bm\mu^{(k)}=\bm\mu
		+(\sqrt{\bm C})^{-T}~(\bm\Gamma^{-1})^k~(\sqrt{\bm C})^{T}\big({\bm \mu}^{(0)}-\bm\mu\big),
	\end{equation}
	where we write
	\begin{equation*}
		\bm \mu=(\sqrt{\bm C})^{-T}\big(P(\widetilde{\bm R}) + \widetilde{\bm B}\big)^{-1}(\sqrt{\bm C})^{T}{\bm m}.
	\end{equation*}
	On the other hand,
	\begin{equation}\label{eq:cov_z}
		\bm\Sigma^{(k)}= (\sqrt{\bm C})^{-T}\left((\bm\Gamma^{-1})^k~(\sqrt{\bm C})^{T}\bm\Sigma^{(0)}(\sqrt{\bm C})\left((\bm\Gamma^{-1})^k\right)^T 
        +\sum_{i=1}^{k} (\bm\Gamma^{-1})^i~{f}_{\delta t}^2(\widetilde{\bm R})\left((\bm\Gamma^{-1})^i\right)^T\right)(\sqrt{\bm C})^{-1}.
	\end{equation}

\end{prop}

The proof is available in Appendix~\ref{proof:mean_cov_ad}.

\vspace{1em}
\begin{rem}
	If the initial condition is chosen so that $\bm \mu^{(0)}=\bm\mu$, then, for any $k\in\lbrace 0, \dots, K\rbrace$, $\e[\bm z^{(k)}]=\bm\mu$, i.e. the expectation of the spatial trace of the field remains constant through time.
\end{rem}

We now move on to the stability result.

\begin{prop}\label{prop:mean_cov_inf}
	Following the notations introduced in \Cref{prop:rec_ad}, let $\bm\mu^{(k)} = \e[\bm z^{(k)}]$ and let $\bm\Sigma^{(k)}=\var[\bm z^{(k)}]$ be respectively the mean and the covariance matrix of the vector $\bm z^{(k)}$, for any $k\in\N_0$. 
	
	If the matrix $\tilde{\bm G} = \frac{1}{2}(\tilde{\bm B}+\tilde{\bm B}^T)$ satisfies $C_P+\lambda_{\min}(\tilde{\bm G}) > 0$, then  the recursion in \Cref{prop:rec_ad} is stable in the sense that the sequences $\lbrace \bm\mu^{(k)} \rbrace_{k\in\N_0}$ and $\lbrace \bm\Sigma^{(k)} \rbrace_{k\in\N_0}$ are bounded. Moreover, both sequences converge to the following limits
	\begin{equation}\label{eq:lim}
		\lim_{k\rightarrow \infty}\bm\mu^{(k)}=\bm\mu, 
		\quad\text{and}\quad \lim_{k\rightarrow \infty} \bm\Sigma^{(k)} = (\sqrt{\bm C})^{-T}\bigg(\sum_{i=1}^{\infty} \bm\Gamma^{-i}{f}_{\delta t}^2(\widetilde{\bm R})(\bm\Gamma^{-i})^T\bigg)(\sqrt{\bm C})^{-1}.
	\end{equation}
\end{prop}

\begin{proof}
	Let $\bm x\in\R^n$. Note that $ \bm x^T \bm\Gamma \bm x = \bm x^T \bm\Gamma^T \bm x = \frac{1}{2}\bm x^T (\bm\Gamma+\bm\Gamma^T) \bm x $. Hence,
	\begin{equation*}
		\bm x^T \bm\Gamma \bm x = \bm x^T \bigg(\frac{1}{2}(\bm\Gamma+\bm\Gamma^T) \bigg)\bm x
		=\bm x^T  \bigg(\bm I  +\frac{\delta t}{c} \big(P(\widetilde{\bm R}) + \widetilde{\bm G}\big) \bigg) \bm x
	\end{equation*}
	since the matrix  $P(\widetilde{\bm R})$ is symmetric. First, note that  the min-max principle gives that  $\bm x^T\widetilde{\bm G}\bm x \ge \lambda_{\min}(\widetilde{\bm G})\Vert \bm x\Vert^2$.  Note then that since the matrix $\widetilde{\bm R}$ is positive semi-definite, its eigenvalues are lower-bounded by zero. Therefore, given that $P$ is lower-bounded by $C_P$ over $\R_+$, we can deduce that $\lambda_{\min}(P(\widetilde{\bm R})) \ge C_P$. Finally, using the min-max principle, we can deduce that $\bm x^TP(\widetilde{\bm R})\bm x \ge \lambda_{\min}(P(\widetilde{\bm R}))\Vert \bm x\Vert^2 \ge C_P\Vert\bm x\Vert^2$. In conclusion, we have for any $\bm x\in\R^n$,
	\begin{equation*}
		\bm x^T \bm\Gamma \bm x 
		\ge \alpha\Vert\bm x\Vert^2
	\end{equation*}
	where $\alpha=\bigg(1  +\frac{\delta t}{c} \big(C_P +\lambda_{\min}(\widetilde{\bm G})\big) \bigg) >1$. 
	
	Then, for any  $\bm x\in\R^n$,
	\begin{equation*}
		\alpha \Vert \bm\Gamma^{-1} \bm x\Vert^2 \le \bm x^T\bm\Gamma^{-T}\bm \Gamma\bm\Gamma^{-1}\bm x =\bm x^T\bm\Gamma^{-T}\bm x = \bm x^T\bm\Gamma^{-1}\bm x
		\le \Vert \bm x\Vert \Vert\bm\Gamma^{-1}\bm x\Vert
	\end{equation*}
	which gives $\Vert \bm\Gamma^{-1} \bm x\Vert\le \alpha^{-1}\Vert\bm x\Vert$. Similarly, we have for any $\bm x\in\R^n$, $\alpha \Vert \bm\Gamma^{-T} \bm x\Vert^2 \le \bm x^T\bm\Gamma^{-T}\bm x\le \Vert \bm x\Vert \Vert\bm\Gamma^{-T}\bm x\Vert$ which gives $\Vert \bm\Gamma^{-T} \bm x\Vert\le \alpha^{-1}\Vert\bm x\Vert$.
	
	Consequently, we have for any $k\in\N_0$, 
	\begin{align*}
		\Vert\bm\mu^{(k)}-\bm\mu\Vert 
		&=\Vert(\sqrt{\bm C})^{-T}\bm\Gamma^{-k}(\sqrt{\bm C})^{T}\big({\bm \mu}^{(0)}-\bm\mu\big)\Vert
		\le \Vert (\sqrt{\bm C})^{-T} \Vert \cdot \Vert\bm\Gamma^{-k}(\sqrt{\bm C})^{T}\big({\bm \mu}^{(0)}-\bm\mu\big) \Vert \\
		&\le \Vert (\sqrt{\bm C})^{-T} \Vert \cdot \alpha^{-k}\Vert(\sqrt{\bm C})^{T}\big({\bm \mu}^{(0)}-\bm\mu\big) \Vert
		\le \alpha^{-k}\Vert (\sqrt{\bm C})^{-T} \Vert \cdot \Vert(\sqrt{\bm C})^{T}\Vert \cdot\Vert{\bm \mu}^{(0)}-\bm\mu \Vert
	\end{align*}
	Since $\alpha >1$, we have $\alpha^{-k} \in (0,1)$ for any $k\in\N_0$ and therefore the sequence $\lbrace \bm\mu^{(k)} \rbrace_{k\in\N_0}$ is bounded. Moreover, by taking the limit as $k\rightarrow\infty$ on both sides of the inequality, we get $\lim_{k\rightarrow\infty} \Vert\bm\mu^{(k)}-\bm\mu\Vert =0$ and therefore $\lim_{k\rightarrow\infty}\bm\mu^{(k)}=\bm\mu$.
	
	For $k\in\N$, let then $\bm X^{(k)}=\bm\Gamma^{-k}(\sqrt{\bm C})^{T}\bm\Sigma^{(0)}(\sqrt{\bm C})(\bm\Gamma^{-k})^T$ and $\bm Y^{(k)}=\sum_{i=1}^{k} \bm\Gamma^{-i}{f}_{\delta t}^2(\widetilde{\bm R})(\bm\Gamma^{-i})^T$. We then have $\bm\Sigma^{(k)}=(\sqrt{\bm C})^{-T}(\bm X^{(k)}+\bm Y^{(k)})(\sqrt{\bm C})^{-1}$. For any $\bm x\in\R^n$ we have
	\begin{align*}
		\Vert \bm X^{(k)}\bm x\Vert 
		&= \Vert \bm\Gamma^{-k}(\sqrt{\bm C})^{T}\bm\Sigma^{(0)}(\sqrt{\bm C})(\bm\Gamma^{-k})^T\bm x\Vert
		\le \alpha^{-k} \Vert (\sqrt{\bm C})^{T}\bm\Sigma^{(0)}(\sqrt{\bm C})(\bm\Gamma^{-k})^T\bm x \Vert \\
		&\le \alpha^{-k} \Vert (\sqrt{\bm C})^{T}\bm\Sigma^{(0)}(\sqrt{\bm C})\Vert \Vert(\bm\Gamma^{-k})^T\bm x \Vert 
		\le \alpha^{-2k} \Vert (\sqrt{\bm C})^{T}\bm\Sigma^{(0)}(\sqrt{\bm C})\Vert \Vert\bm x \Vert,
	\end{align*}
	meaning in particular that $\Vert \bm X^{(k)}\Vert \le \alpha^{-2k} \Vert (\sqrt{\bm C})^{T}\bm\Sigma^{(0)}(\sqrt{\bm C})\Vert$. Similarly, for any $\bm x\in\R^n$, we have
	\begin{align*}
		\Vert \bm Y^{(k)}\bm x\Vert 
		&= \Vert \sum_{i=1}^{k} \bm\Gamma^{-i}{f}_{\delta t}^2(\widetilde{\bm R})(\bm\Gamma^{-i})^T\bm x\Vert
		\le  \sum_{i=1}^{k} \Vert\bm\Gamma^{-i}{f}_{\delta t}^2(\widetilde{\bm R})(\bm\Gamma^{-i})^T\bm x\Vert 
		\le \sum_{i=1}^{k} \alpha^{-2i}\Vert{f}_{\delta t}^2(\widetilde{\bm R})
		\Vert \Vert \bm x\Vert, 
	\end{align*}
	meaning in particular that 
	\begin{align*}
		\Vert \bm Y^{(k)}\Vert 
		\le \Vert{f}_{\delta t}^2(\widetilde{\bm R})
		\Vert\sum_{i=1}^{k} \alpha^{-2i}  = \frac{1-\alpha^{-2k}}{1-\alpha^{-2}} \Vert{f}_{\delta t}^2(\widetilde{\bm R})
		\Vert. 
	\end{align*}
	We can then conclude that the sequences $\lbrace \bm X^{(k)}\rbrace_{k\in\N_0}$ and $\lbrace \bm Y^{(k)}\rbrace_{k\in\N_0}$ are both bounded, and therefore that the sequence $\lbrace \bm \Sigma^{(k)}\rbrace_{k\in\N_0}$ is bounded. Besides, by taking the limit as $k\rightarrow\infty$, we get that $\Vert \bm X^{(k)}\Vert\rightarrow 0$ and that the matrix series $\bm Y^{(k)}$ converges (absolutely). Hence, we end up with
	\begin{equation*}
		\lim_{k\rightarrow \infty} \bm\Sigma^{(k)} 
		= (\sqrt{\bm C})^{-T}\bigg(\lim_{k\rightarrow \infty} \bm X^{(k)} +\lim_{k\rightarrow \infty} \bm Y^{(k)} \bigg)(\sqrt{\bm C})^{-1}
		= (\sqrt{\bm C})^{-T}\bigg(\sum_{i=1}^{\infty} \bm\Gamma^{-i}{f}_{\delta t}^2(\widetilde{\bm R})(\bm\Gamma^{-i})^T\bigg)(\sqrt{\bm C})^{-1}.
	\end{equation*}
	
\end{proof}

Hence, whenever the smallest eigenvalue of the symmetric part of the matrix $\widetilde{\bm B}$ is strictly lower-bounded by $-C_P$, no instability occurs. It turns out that when a divergence-free vector field is considered in SPDE~\eqref{eq:spde}, the recursion is guaranteed to be stable, as stated in the next result.

\begin{corol}
    When considering a time-invariant and divergence-free vector field $\gamma=\gamma_t$ in SPDE~\eqref{eq:spde}, the recursion in \Cref{prop:rec_ad} is stable and satisfies~\eqref{eq:lim}.
\end{corol}

\begin{proof}
    To prove this result, we simply apply \Cref{prop:mean_cov_inf} after noticing that, for a divergence-free vector field, $\tilde{\bm G}=0$.

    Indeed, Let $ 1\le i,j\le n$. On the one hand, using the Leibnitz rule, we have 
	\begin{align*}
		\langle\psi_j^\ell,\dv(\gamma\psi_i^\ell)\rangle
		=\langle\psi_j^\ell,\dv(\gamma)\psi_i^\ell\rangle+\langle\psi_j^\ell,g(\gamma,\nabla\psi_i^\ell)\rangle
        =\langle\psi_j^\ell,g(\gamma,\nabla\psi_i^\ell)\rangle.
	\end{align*}
	On the other hand, the integration by parts formula gives
		\begin{align*}
		\langle\psi_i^\ell,\dv(\gamma\psi_j^\ell)\rangle
		=-\langle\nabla\psi_i^\ell,\gamma\psi_j^\ell\rangle
		=-\langle g(\nabla\psi_i^\ell,\gamma),\psi_j^\ell\rangle.
	\end{align*}
    Hence, we can write, for any $ 1\le i,j\le n$,
	\begin{align*}
	2 {G}_{ij}
	=\langle\psi_i^\ell,\dv(\gamma\psi_j^\ell)\rangle+\langle\psi_j^\ell,\dv(\gamma\psi_i^\ell)\rangle
	=0,
\end{align*}
and therefore $\bm G =0$, which in turn implies that $\tilde{\bm G}=0$ and $\lambda_{\min}(\tilde{\bm G})=0$.
\end{proof}

\begin{rem}
    If we no longer consider that the vector field $\gamma$ is divergence-free, it is possible to show (cf. \Cref{prop:stab})  that the condition  $C_P+\lambda_{\min}(\tilde{\bm G}) > 0$ in \Cref{prop:mean_cov_inf} is satisfied whenever  $\dv(\gamma)$ satisfies the condition 
\begin{equation*}
	\inf_{\mathcal{M}} \bigg(C_P +  \frac{1}{2}\dv(\gamma)\bigg)  > 0.
\end{equation*}  
Note that this last inequality is one of the conditions we assumed for the well-posedness of the SPDE~\eqref{eq:spde}.

However, in practice, when using methods like Surface Finite Elements~\cite{bonito2022numerical}, we replace the surface integrals on $\mathcal{M}$  by integrals computed on its polyhedral approximation $\mathcal{M}_h$ when computing the finite element matrices $\bm C$, $\bm R$ and $\bm B$ (according to~\eqref{eq:coef-approx}). When doing so, differences of order $\mathcal{O}(h^2)$ (where $h$ is the size of the mesh) are to be expected when comparing both types of integrals. Consequently, the stability of the recursion is no longer guaranteed. However, by decreasing the size of the mesh, we reduce the discrepancy between the two types of integrals, and therefore we can retrieve stability (which is guaranteed at the limit $h\rightarrow 0$). In practice, though, we only need to decrease the mesh size until the matrix $\widetilde{\bm G}$ in \Cref{prop:mean_cov_inf} satisfies the inequality $C_P+\lambda_{\min}(\tilde{\bm G}) > 0$.
\end{rem}

Generalizing \Cref{prop:mean_cov_ad,prop:mean_cov_inf} to the case where the advection vector field is time-dependent is straightforward. Indeed, the same arguments as the ones used in the proof of \Cref{prop:mean_cov_ad} can be used to show that covariance matrix $\bm\Sigma^{(k)}$ can be expressed using \eqref{eq:cov_z}, but replacing the products $(\bm\Gamma^{-1})^k$ by $\prod_{i=1}^{k}((\Gamma^{(k-i)})^{-1})$. The stability of the covariance can then be obtained assuming that for any $k$, the matrix $\tilde{\bm G}^{(k)} = \frac{1}{2}(\tilde{\bm B}^{(k)}+(\tilde{\bm B}^{(k)})^{T})$ satisfies $C_P+\lambda_{\min}(\tilde{\bm G}^{(k)}) \ge C_G$ for some $C_G>0$. 

As for the stability of the mean, the extension is less trivial but can be tackled using  stability results of linear difference equations (cf. eg. \citep{elaydi2005introduction}). A particular where the stability is ensured though, is the case where no source term is considered (i.e. $\bm m=\bm 0$). Then, the same requirement on the matrices $\tilde{\bm G}^{(k)}$ as the one given above suffices to ensure stability.

\section{Prediction and inference from data}\label{sec:inf_pred}

In this section, we describe the algorithms used for inference, prediction, and conditional simulation of the spatio-temporal field from observed data.

Let $T>0$. We consider a spatio-temporal process $u(t, s)$ defined on the domain $[0, T] \times \mathcal{M}$, where $\mathcal{M}$ is the spatial surface of interest. Our goal is to model observations of $u$ using a combination of fixed effects (regression on covariates) and random effects representing the underlying spatio-temporal dynamics, modeled as the solution to an advection-diffusion SPDE. An independent measurement noise is added to account for observation error.

More formally, the continuous model can be written as:
$$u(t, s) = \bm \eta(t,s)^T \bm b + Z(t, s) + \epsilon(t, s),$$
where $\bm \eta(t,s)=(\eta_1(t,s), \dots, \eta_q(t,s))^T$ is the vector of covariate fields, $\bm b$ is the vector of the $q$ fixed effect coefficients, $Z(t,s)$ is the random field solution of the SPDE, and $\epsilon(t,s)$ is an independent Gaussian noise with variance $\sigma^2$.

For computational purposes, we discretize the time interval $[0,T]$ into $K+1$ regular time steps of size $\delta t=T/K$ and the manifold $\mathcal{M}$ into the polyhedral surface $\mathcal{M}_h$. We hence write $t_k=k\delta t$ for $k\in \lbrace 0, \dots, K\rbrace$. Let $n_k$ be the number of observations at time $t_k$, located at points  $s_1^{(k)},\dots,s_{n_k}^{(k)}\in\mathcal{M}_h$. Denoting by $\bm u^{(k)}=(u\big(t_k,s_1^{(k)}\big),\dots,u\big(t_k,s_{n_k}^{(k)}\big))^T\in \R^{n_k}$ the vector of observations at time $t_k$, the full observation vector containing all the $N_o= \sum_{k=0}^K n_k$ observations is $\bm U =(\big(\bm u^{(0)}\big)^T, \dots, \big(\bm u^{(K)}\big)^T)^T\in\R^{N_o}$.
Let $\bm\varepsilon$ be a vector of $N_o$ independent standard Gaussian variables and $\sigma>0$. Then the statistical model for the observations takes the form
\begin{equation}\label{eq:obs}
	\bm U = \bm \eta \bm b + \bm A^T \bm Z + \sigma \bm\varepsilon, 
\end{equation}
where $\bm b\in\R^{q}$ is the vector of $q$ fixed effects, $\bm\eta \in \R^{N_o\times q}$ is a matrix of covariates, $\bm Z$ is the vector containing the weights defining the solution of SPDE~\eqref{eq:spde_discr} as in \Cref{prop:Qeuler}, and $\bm A \in \R^{(K+1)N\times N_o}$ is the block diagonal matrix whose $k$-th block ($k\in \lbrace 0, \dots, K\rbrace$)  $\bm A^{(k)}\in\R^{N\times n_k}$ is defined by
\begin{equation*}
	[\bm A^{(k)}]_{ij} =\psi_i(s_j^{(k)}), \quad i\in \lbrace 1, \dots, N\rbrace,\; j\in \lbrace 1, \dots, n_k\rbrace.
\end{equation*}

\subsection{Inference}\label{sec:inf}

In this section we discuss the parameter inference strategy in two different cases, i.e., when data are available at all mesh locations and all temporal steps (we call it "fully observed field") and when data are scattered in space but available at all time steps (we call it "partially observed field").

\subsubsection{Fully observed field}\label{sec:fo}

Let $\bm\theta$ be the vector containing all the parameters that parametrize~\eqref{eq:spde_discr}, and let $\bm\nu=(\bm\theta^T,\bm \mu^{(0)},\bm m^T,\bm b^T,\sigma^2)^T$ be the vector containing all the parameters of the statistical model. Let us first assume that we observe the field $Z$ exactly at all mesh locations, i.e. $\bm U = \bm Z$. In that case, the log-likelihood of the parameters is simply given by
\begin{equation}
	\begin{aligned}
		\mathcal{L}(\bm\nu)=-\frac{N}{2}\log 2\pi
		+\frac{1}{2}\log\vert \bm Q_{\bm Z}(\bm\theta) \vert
		-\frac{1}{2}(\bm Z-\bm \mu_{\bm Z}(\bm\theta))^T \bm Q_{\bm Z}(\bm\theta)(\bm Z-\bm \mu_{\bm Z}(\bm\theta)),
	\end{aligned}
	\label{eq:loglik}
\end{equation}
where the expressions of $\bm \mu_{\bm Z}$ and $\bm Q_{\bm Z}$ are given in~\Cref{prop:Qeuler}. 

This case, where observations are available at all mesh nodes, can be viewed as a calibration setting that facilitates identification of the relationship between the advection field and the variable of interest. The resulting model can then be used in more realistic scenarios with sparse observations (e.g., monitoring stations) to reconstruct or predict the field over the entire mesh.

The log-determinant $\bm Q_{\bm Z}$ in \eqref{eq:loglik} can be deduced from the log-determinants of the diagonal block entries of the matrices $\bm L(\bm \Gamma)$, $\bm D\big(f_0^{-2}(\widetilde{\bm R}),\, {f}_{\delta t}^{-2}(\widetilde{\bm R})\big)$ and $\bm D\big(\sqrt{\bm C}\big)$. We obtain in particular the relations
\begin{equation*}
	\begin{aligned}
		\log\vert \bm Q_{\bm Z}\vert 
		&=2(K+1)\log\vert \sqrt{\bm C}\vert +\sum_{k=0}^{K-1}\log\vert (\bm \Gamma^{(k)})^T({\bm \Gamma}^{(k)})\vert +\log\vert f_0^{-2}(\widetilde{\bm R})\vert+K\log\vert{f}_{\delta t}^{-2}(\widetilde{\bm R})\vert, \\
		&=2(K+1)\log\vert \sqrt{\bm C}\vert +\sum_{k=0}^{K-1}\log\vert (\bm \Gamma^{(k)})^T({\bm \Gamma}^{(k)})\vert  +2\log\vert f_0^{-1}(\widetilde{\bm R})\vert+2K\log\vert{f}_{\delta t}^{-1}(\widetilde{\bm R})\vert
	\end{aligned}
\end{equation*}
Hence, in order to compute the log-determinant, we do not need to build and store the matrix $\bm Q_{\bm Z}$, but rather only need to store the matrices $\widetilde{\bm R}$, $\sqrt{\bm C}$ and $\bm \Gamma^{(k)}$, for $k=0,\dots,K-1$. Indeed, the log-determinants $\log\vert (\bm \Gamma^{(k)})^T({\bm \Gamma}^{(k)})\vert$, $\log\vert{f}_{0}^{-1}(\widetilde{\bm R})\vert=-\log\vert{f}_{0}(\widetilde{\bm R})\vert$ and $\log\vert{f}_{\delta t}^{-1}(\widetilde{\bm R})\vert=-\log\vert{f}_{\delta t}(\widetilde{\bm R})\vert$ are approximated using a stochastic trace estimator (Hutchinson estimator, cf.~\cite{pereira2022geostatistics}) combined with the Lanczos algorithm~\cite{lanczos1950}, which efficiently approximates quadratic forms involving $\log(A)$ using only matrix–vector products. This method, originally proposed in \cite{dong2017scalable}, ensures that the storage needs for that computation remain the same as the number of time steps considered $K$ increases (which is not the case when using directly the matrix $\bm Q_{\bm Z}\in\R^{(K+1)N\times(K+1)N}$).

As for the quadratic form, we first solve the linear system $\bm Q_{\bm Z} \bm X = \bm Y$ iteratively by substitution, exploiting the bi-diagonal structure of the matrix $\bm L\left(\{\bm \Gamma^{(k)}\}_{k=0}^{K-1}\right)$. This approach requires only the ability to compute products between functions of the matrix $\widetilde{\bm R}$ and vectors and to solve linear systems with the matrices $\bm \Gamma^{(k)}$ (cf.~\Cref{alg:linsolveQ}). Afterwards, we only need to compute matrix--vector products involving $\bm Q_{\bm Z}$ (cf.~\Cref{alg:mvprodQ}).

In the likelihood presented above, the initial condition is assumed to be sampled from~\eqref{eq:ic}. As this assumption might not be justified in practice, we present below a framework which allows us to infer the model parameters on a modified likelihood. The idea is to consider the joint distribution of the time steps $(\bm z^{(1)}, \dots, \bm z^{(K)})$ conditioned on the initial condition $\bm z^{(0)}$, rather than the full joint distribution of $\bm Z = (\bm z^{(0)},\bm z^{(1)}, \dots,\bm z^{(K)})^T$. In that way, no modeling assumption is needed on the initial condition $\bm z^{(0)}$, but only the parameters linked to the SPDE dynamics (i.e. all but the parameters linked to the initial condition) can be inferred.

Indeed, given $\bm \Theta \in\R^{N\times N}$ and $\lbrace\bm \Theta^{(k)}\rbrace_{k=0}^{K-1}\subset \R^{N\times N}$, let $\bm L_0\left(\lbrace\bm \Theta^{(k)}\rbrace_{k=0}^{K-1}\right)$, $\bm D(\bm \Theta)\in~\R^{KN\times K N}$ be the block matrices  given by
\begin{equation}
	\bm L_0\left(\lbrace\bm \Theta^{(k)}\rbrace_{k=0}^{K-1}\right)
	=\begin{smallmat}
\bm\Theta^{(0)} & & &\\
-\bm{I} & \bm\Theta^{(1)} & &\\
&\ddots & \ddots &\\
& & -\bm{I} & \bm\Theta^{(K-1)}
\end{smallmat}, \quad 
	\bm D_0(\bm \Theta)
	=\begin{smallmat}
		{\bm \Theta} & &     \\
		&   \ddots  &   \\
		&     & {\bm \Theta}\\
	\end{smallmat}.
	\label{eq:defL_new}
\end{equation}

Noting that, following the recursion (and notations) of \Cref{prop:rec_ad},
\begin{equation}
\bm L_0\left(\lbrace\bm \Gamma^{(k)}\rbrace_{k=0}^{K-1}\right)
\begin{pmatrix}
\bm x^{(1)}\\
\vdots\\
\bm x^{(K)}
\end{pmatrix}
=
\begin{pmatrix}
\bm x^{(0)}+\frac{\delta t}{c}(\sqrt{\bm C})^T \bm m\\
\frac{\delta t}{c}(\sqrt{\bm C})^T \bm m\\
\vdots \\
\frac{\delta t}{c}(\sqrt{\bm C})^T \bm m
\end{pmatrix}
+
\bm D_0({f}_{\delta t}(\widetilde{\bm R}))
\begin{pmatrix}
\bm w^{(1)}\\
\vdots\\
\bm w^{(K)}
\end{pmatrix},
\end{equation}
we can conclude that the distribution of $(\bm z^{(1)}, \dots, \bm z^{(K)}\vert \bm z^{(0)})$ is also multivariate Gaussian with mean $\bm\mu_0$ and precision matrix $\bm Q_0$ given by
\begin{align*}
    \bm\mu_0 &= \bm D_0\left((\sqrt{\bm C})^{-T}\right)\bm L_0\left(\lbrace\bm \Gamma^{(k)}\rbrace_{k=0}^{K-1}\right)^{-1}\bm D_0\left((\sqrt{\bm C})^T\right)\begin{pmatrix}
\bm z^{(0)}+\frac{\delta t}{c} \bm m\\
\frac{\delta t}{c}  \bm m\\
\vdots \\
\frac{\delta t}{c} \bm m
\end{pmatrix}, \\
\bm Q_0 &= \bm D_0\left((\sqrt{\bm C})\right)\bm L_0\left(\lbrace\bm \Gamma^{(k)}\rbrace_{k=0}^{K-1}\right)^{T} \bm D_0({f}_{\delta t}^{-2}\left(\widetilde{\bm R})\right) \bm L_0\left(\lbrace\bm \Gamma^{(k)}\rbrace_{k=0}^{K-1}\right)\bm D_0\left((\sqrt{\bm C})^{T}\right).
\end{align*}
The associated likelihood takes the same form as~\eqref{eq:loglik}, but where $\bm\mu_{\bm Z}$ is replaced by $\bm\mu_0$ and $\bm Q_{Z}$ by $\bm Q_0$.

\subsubsection{Partially observed field}\label{sec:inf_part}

The next proposition, proven in \cite[Section 3.1]{clarotto2024spde}, gives the log-likelihood in the more general setting of scattered observations in space $\bm U$ given by~\eqref{eq:obs}. 

From now on, for any $s>0$ and $\bm Q \in\R^{(K+1)N\times (K+1)N}$ a positive definite matrix, let us denote by $\mathcal{K}(\cdot \,\vert\, \bm Q,s^2)$ the linear map given by
\begin{equation}
	\mathcal{K}(\bm v \,\vert\,\bm Q,s^2)
	=s^{-2}(\bm Q + s^{-2}\bm A\bm A^T)^{-1}\bm A \bm v, \quad \bm v\in\R^{N_o}.
	\label{eq:defK}
\end{equation}

\begin{prop}
	The log-likelihood function of the vector of observations $\bm U$ is given by
	\begin{equation}
		\begin{aligned}
		\mathcal{L}(\bm\nu)=-\frac{N_o}{2}\log 2\pi
		&+\frac{1}{2}\log\vert \bm Q_{\bm U}(\bm\nu) \vert
		-\frac{\sigma^{-2}}{2}
		\Vert \bm U -\bm A^T \bm \mu_{\bm Z}- \bm\eta\bm b\Vert^2_2\\
		&+\frac{\sigma^{-2}}{2}(\bm U -\bm A^T \bm \mu_{\bm Z}- \bm\eta\bm b)^T\bm A^T\mathcal{K}\big(\bm U -\bm A^T \bm \mu_{\bm Z}- \bm\eta\bm b \,\vert\, \bm Q_{\bm Z}(\bm\theta),\sigma^2\big),
		\end{aligned}
		\label{eq:loglik_part}
	\end{equation}
	where $\bm Q_{\bm U}(\bm\nu)=(\bm A^T \bm Q_{\bm Z}(\bm\theta)^{-1}\bm A + \sigma^2\bm I)^{-1}$, and
	\begin{equation*}
		\log\vert \bm Q_{\bm U}(\bm\nu) \vert=-N_o\log\sigma^2+\log\vert \bm Q_{\bm Z}(\bm\theta) \vert-
		\log\vert \bm Q_{\bm Z}(\bm\theta) +\sigma^{-2}\bm A\bm A^T \vert.
	\end{equation*}
	
\end{prop}

Therefore, efficient algorithms to evaluate $\mathcal{K}(\cdot \,\vert\, \bm Q,s^2)$ and $\log\vert \bm Q_{\bm U}(\bm\nu)\vert$ are fundamental to evaluate the log-likelihood. A first approach to compute~\eqref{eq:defK} consists in factorizing the matrix $\bm \Psi=(\bm Q + s^{-2} \bm A\bm A^T )$ (using for instance a Cholesky decomposition), and then using the factorization to efficiently compute the log-likelihood.  However, since the matrix $\bm \Psi$ has size $N(K+1)\times N(K+1)$, building, storing and factorizing can become computationally prohibitive in some applications where either $N$ or $K$ (or both) are large.

In this paper, an alternative approach is used to evaluate $\mathcal{K}(\cdot \,\vert\, \bm Q,s^2)$. It consists in solving the linear system in~\eqref{eq:defK} using a matrix-free iterative algorithm~\cite{saad2003iterative}. Such algorithms yield an approximate solution of the linear system through an iterative process which requires at each iteration products between $\bm \Psi$ and vectors (see, for example, \cite[Section 4]{pereira2022geostatistics} and \cite[Section 3.1]{clarotto2024spde}). Such products can in turn be evaluated without having to explicitly build the matrix $\bm \Psi$, but using instead the diagonal block structure of $\bm A$. In this setting, only the \q{spatial} matrices $\widetilde{\bm R}$, $\sqrt{\bm C}$, $\bm \Gamma^{(k)}$ ($0\le k\le K-1$) and $\bm A^{(k)}$ ($0\le k\le K$), which are sparse and of size at most $N\times N$, are stored. In this paper, the Generalized minimal residual method or GMRES (see, for example, \cite{ascher2011}) is used. It is an iterative Krylov method based on the Lanczos algorithm~\cite{lanczos1950} which only requires matrix-vector products. To improve its convergence, a block Gauss-Seidel symmetric preconditioner is used (cf. \Cref{alg:blockGSPrecond}).
The cost of this system resolution will be exactly the same that will be needed for a complete spatio-temporal prediction by kriging, outlined in section \ref{sec:krig}.

The log-determinant terms are evaluated using matrix-free approaches based on Hutchinson trace estimators \cite{han2015large}. The maximization of the log-likelihood is tackled using the Nelder-Mead algorithm \cite{nelder1965simplex}, which only requires evaluations of the cost function.

\subsection{Prediction by kriging}\label{sec:krig}

Predictions of the spatio-temporal field are tackled using conditional expectations (and variances), following the approach outlined by \citet{clarotto2024spde,pereira2022geostatistics}. 
The next proposition provides explicit formulas for the computation of the conditional expectation and variance of the field $\bm Z$ given the observations $\bm U$.

\begin{prop}\label{prop:krig}
	The conditional expectation (also called kriging predictor) $\e[\bm Z\vert \bm U]$ of $\bm Z$ given $\bm U$ is given by
	\begin{equation*}
		\e[\bm Z\vert \bm U]=\bm \mu_{\bm Z}+\mathcal{K}(\bm U -\bm A^T \bm \mu_{\bm Z} - \bm\eta{\bm b} \,\vert\, \bm Q_{\bm Z},\sigma^2),
	\end{equation*}
where the expressions of $\bm \mu_{\bm Z}$ and $\bm Q_{\bm Z}$ are given in~\Cref{prop:Qeuler}.
	 Besides, the conditional variance $\var[\bm Z\vert \bm U]$ is given by
	\begin{equation*}
		\var[\bm Z\vert \bm U]=\e\big[(\bm Z-\e[\bm Z\vert \bm U])(\bm Z-\e[\bm Z\vert \bm U])^T\vert \bm U\big]=(\bm Q_{\bm Z} + \sigma^{-2}\bm A\bm A^T)^{-1}.
	\end{equation*}
\end{prop}

\begin{proof}
	This result is a direct consequence of the fact that the vector $(\bm Z^T, \bm U^T )^T$ is Gaussian, and a complete proof is given in~\cite[Proposition 3.1]{pereira2022geostatistics}.
\end{proof}

	For $k\in\lbrace 0,\dots, K\rbrace$, the spatial prediction $Z^*(t_k,p)$ of the field $Z(t_k,\cdot)$ at any location $p\in\mathcal{M}_h$ is deduced from the conditional expectation $\e[\bm Z\vert \bm U]$ by leveraging the linearity of the (conditional) expectation, thus giving:
\begin{equation}
	Z^*(t_k,p)=\e[Z(t_k,p)\vert \bm U]=\begin{smallmat}
		\psi_1(p) \\
		\vdots \\
		\psi_N(p) 
	\end{smallmat}^T \e[\bm Z\vert \bm U].
	\label{eq:spat_pred}
\end{equation}
Time extrapolation at times $t_k$, $k>K$, can be handled in a similar fashion. Indeed, by taking the conditional expectation $\e[~\cdot~\vert \bm U]$ on both sides of the recursion~\eqref{eq:rec_ad} we get
\begin{equation}
	\e[\bm z^{(k+1)}\vert \bm U]
	=(\sqrt{\bm C})^{-T} \bm(\Gamma^{(k)})^{-1} (\sqrt{\bm C})^T\bigg(\e[\bm z^{(k)}\vert \bm U]+\frac{\delta t}{c}{\bm m}\bigg), \quad  k\ge K,
\end{equation}
where we recall that $\bm z^{(k)}$ is the weight vector defining the solution $Z$ at time $t_k$ (cf. \Cref{prop:rec_ad}), and $\e[\bm z^{(K)}\vert \bm U]$ corresponds to the $N$ last rows of $\e[\bm Z\vert \bm U]$. Then spatial predictions at any locations can once again be obtained using~\eqref{eq:spat_pred}.

Following the same approach outlined in section \ref{sec:inf_part}, we can compute kriging predictions by preconditioned GMRES method.

Finally, note that the conditional expectation can also be used to generate conditional simulations at time steps $t_k$, $0\le k\le K$ of the field $Z$, by leveraging the fact that the conditional variance does not depend explicitly on the conditioning data $\bm U$. This approach,  presented in more details in \cite[Section 3.3]{clarotto2024spde}, is recalled below:
\begin{enumerate}
	\item Compute a non-conditional simulation $\bm Z_{NC}$ by running the recursion \eqref{eq:rec_ad}.
	\item Generate new observations by computing $\bm U_{NC} = \bm \eta \bm b + \bm A^T \bm Z_{NC} + \sigma \bm\varepsilon_{NC}$.
	\item Compute the residuals  $\bm r_{NC}= \bm Z_{NC} - \e[\bm Z_{NC}\vert \bm U_{NC}]$.
	\item Return the conditional simulation $\bm Z_C=\e[\bm Z\vert \bm U] + \bm r_{NC}$.
\end{enumerate}
Conditional simulations at further time steps are then obtained using once again the recursion \eqref{eq:rec_ad}.

\section{Application}

We demonstrate our methodology using data from the ECMWF Atmospheric Composition Reanalysis 4 (EAC4) dataset~\cite{innes2019CAMS}, focusing on a 3-day period from December 1 to December 3, 2024. The dataset provides total aerosol optical depth (AOD) at 550 nm, recorded every 3 hours for a total of 21 time steps on a grid with a spatial resolution of $0.75^\circ \times 0.75^\circ$. Aerosol optical depth measures the extinction of solar radiation by aerosol particles in the atmosphere, with higher AOD values corresponding to heavier aerosol loads. It serves as a key indicator of dust concentration.

As a first preprocessing step, we apply a Box-Cox transformation with $\alpha = 0.15$ to approximate Gaussianity for the AOD values~\cite{box1964transform}. The transformed data are then centered and scaled (\Cref{fig:aod}). We model the transformed AOD data as a realization of a spatio-temporal Gaussian process governed by a source-free advection-diffusion SPDE defined on a Riemannian manifold (here, the sphere), as introduced in the previous sections.  

We discretize the globe using an icosahedron-based uniform triangulation of 163842 nodes. Since the number of observation locations is only slightly larger than the number of mesh points, we assign data to the mesh by associating each node with the value of its nearest observation.

We take $M_S=0$ (which implies that $\bm m=\bm 0$), $\bm \mu^{(0)}=\bm 0$, $P(\lambda)=\kappa^2+\lambda$, $f_S(\lambda)=(\kappa_S^2+\lambda)^{-1}$ and $f_0(\lambda)=(\kappa_{in}^2+\lambda)^{-1}$. To account for advection, we extract the 10m u-component and v-component wind fields from the same EAC4 dataset, using the same temporal and spatial resolution. These components are used as a proxy for the advection field. Since the real wind field is not necessarily divergence-free, in order to ensure the stability of the recursion in~\Cref{prop:rec_ad}, we decompose the field with the Helmholtz-Hodge decomposition (cf. end of \Cref{sec:model}) and only select the divergence-free term for our case study. This implies forgetting about small scales and tropical patterns in the wind flow. The vector field $\gamma_t$ is temporally varying, which implies that the (scaled) advection matrix $\widetilde{\bm B}^{(k)}$ is different for each time step $k$. 
Moreover, a global advection scaling coefficient $c_{adv}$ is multiplied to the advection term to modulate the influence of the wind field and compensate for potential discrepancies in magnitude. We assume a spatio-temporal process with zero mean and an explicit measurement error component with variance $\sigma^2$, without incorporating fixed effects.

\begin{figure}
	\centering
	{\includegraphics[width=0.85\textwidth]{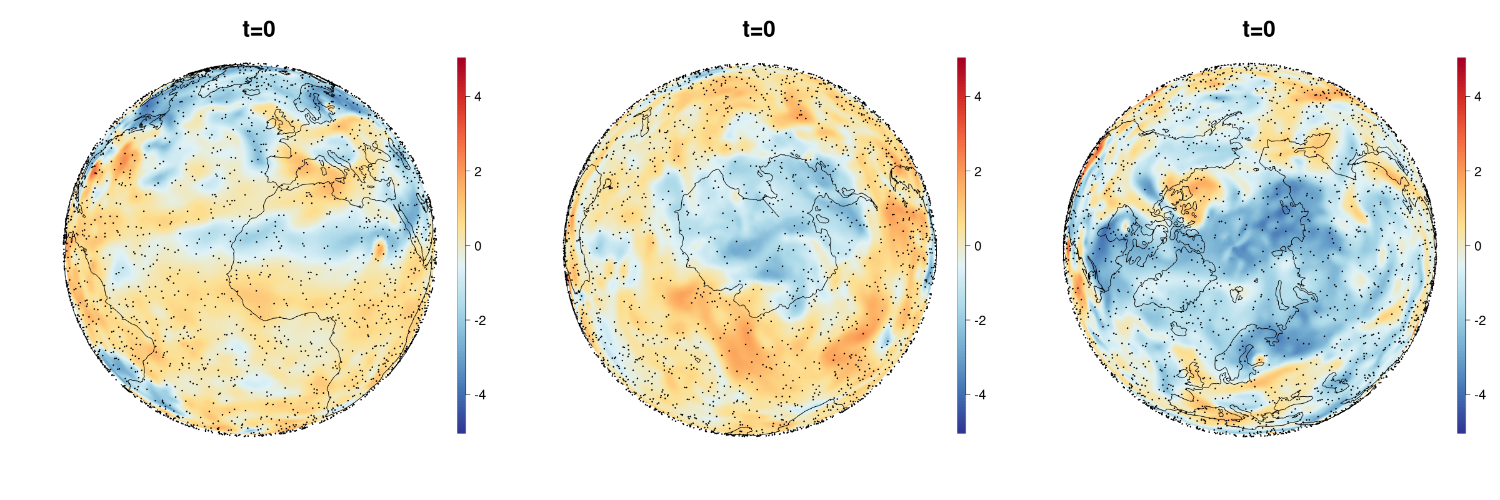}}\\{\includegraphics[width=0.85\textwidth]{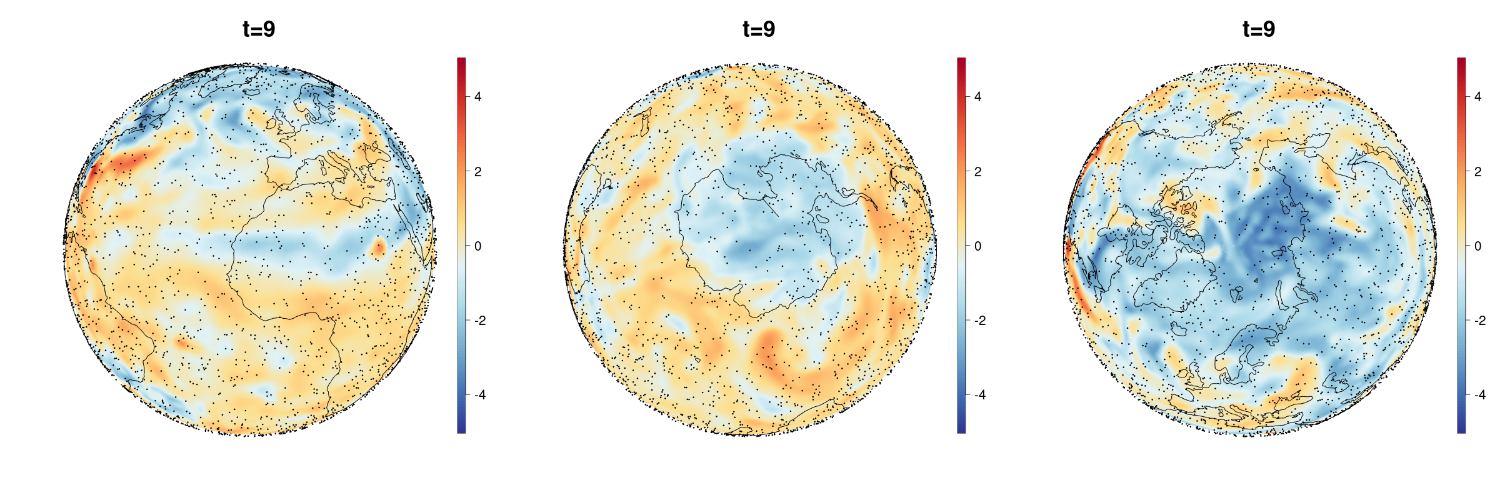}}\\{\includegraphics[width=0.85\textwidth]{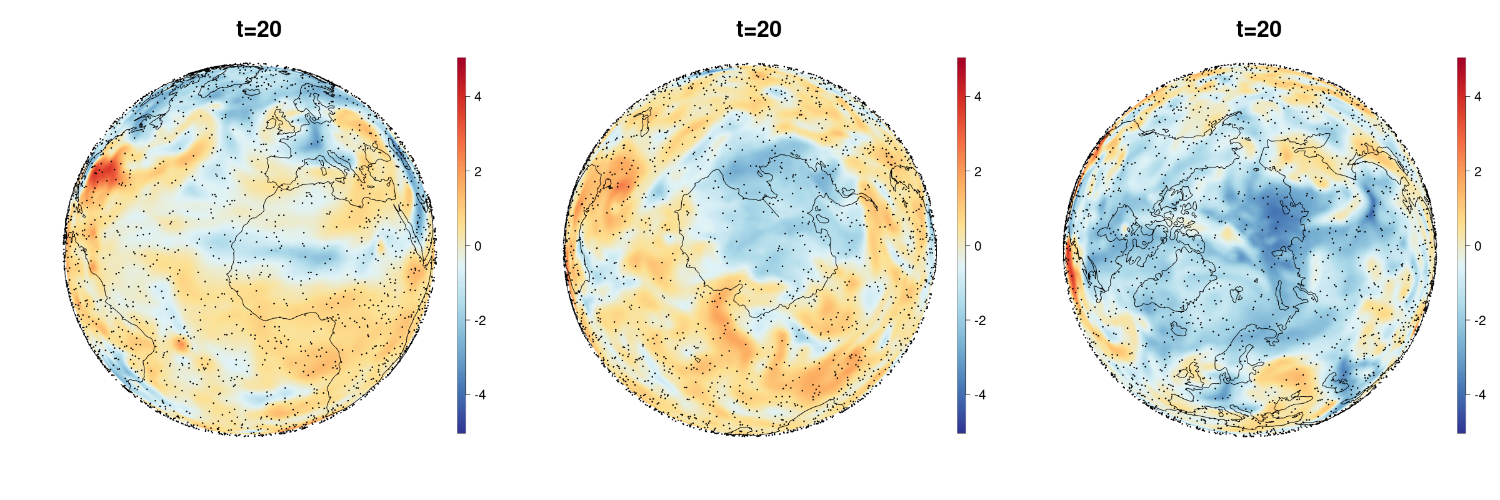}}
	\caption{\link{https://mike-pereira.github.io/PRED_STRF/dat-aod}{AOD data (after preprocessing) on the globe represented from three different viewpoints on the surface at three different time steps.}}\label{fig:aod}
\end{figure}

5000 spatial locations (approximately $3\%$ of the total) are randomly selected among the mesh nodes, and treated as measuring stations at all time steps. Our goal is to estimate the parameters of the SPDE model: the damping coefficient $\kappa^2$, the advection scaling parameter $c_{adv}$, the temporal scaling factor $c$, the range parameter $\kappa^2_{in}$ of the stationary spatial diffusion and unit-variance SPDE model at the initial time step, the range parameter $\kappa_S^2$ and variance scaling $\sigma_S^2$ characterizing the colored noise, and the variance of the noise $\sigma^2$. The vector of parameters $\bm \nu_{\text{AD}} = (\kappa^2,c,c_{adv},\kappa^2_{in},\kappa_S^2,\tau,\sigma)^T$ is inferred by maximizing the Gaussian log-likelihood of the observed data of equation~\eqref{eq:loglik_part}.

To reduce the computational burden of inference, we restrict the observations to the first 11 time steps. This will only be done for the inference step, and predictions will nevertheless be carried out on the 21 available time steps. We use the Nelder-Mead algorithm for the optimization \cite{nelder1965simplex}, which relies only on evaluations of the log-likelihood.

To assess the models' prediction capability, we perform kriging and conditional simulations using the same 5000 observations locations, taken over the 21 available time steps. The kriging results are shown in \Cref{fig:krig_aod}. They clearly illustrate that the advection-diffusion model successfully reproduces the spatio-temporal dynamics of dust transport driven by wind advection and diffusion, even in the presence of a limited number of measurement stations. The smoother appearance of the maps is consistent with the inherent characteristics of classical kriging. 

\begin{figure}
	\centering
	{\includegraphics[width=0.85\textwidth]{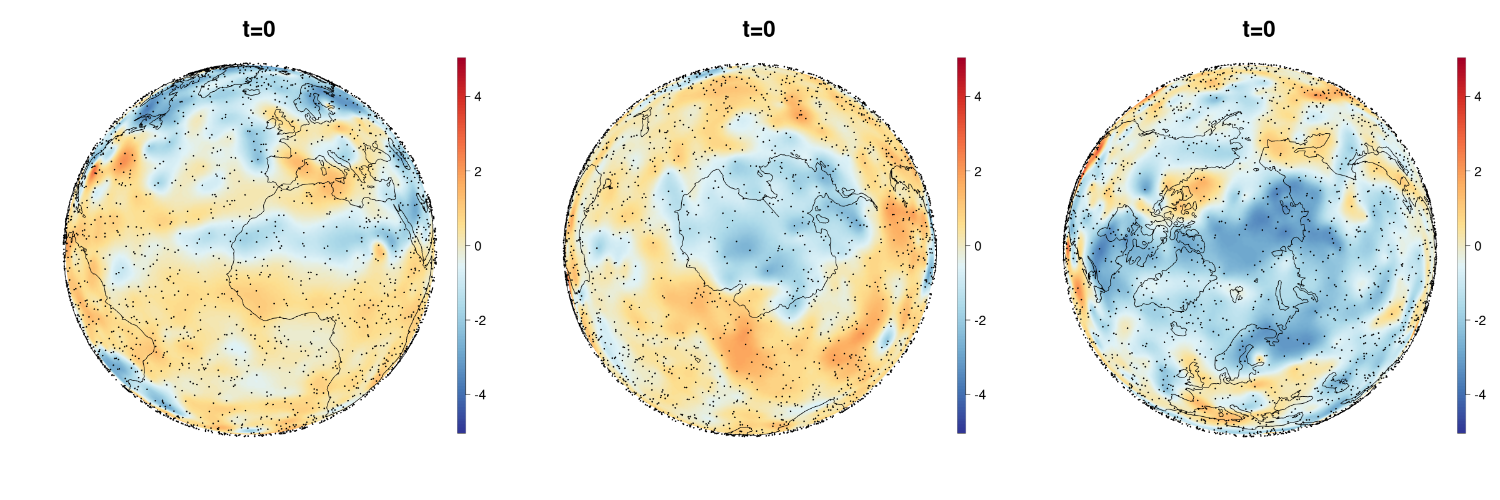}}\\{\includegraphics[width=0.85\textwidth]{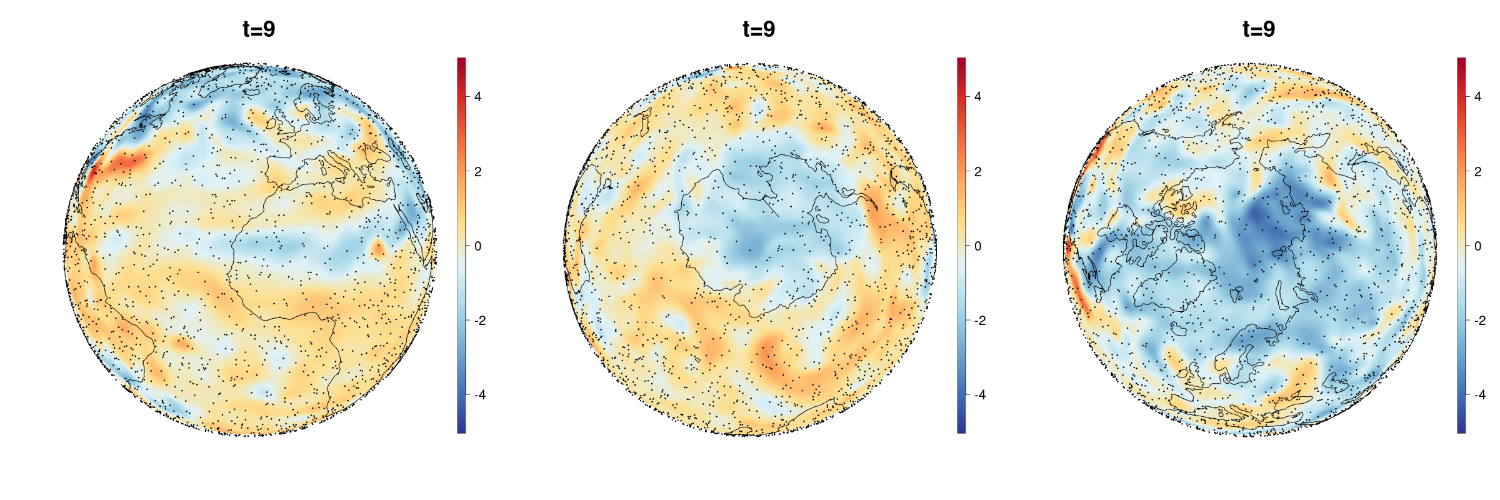}}\\{\includegraphics[width=0.85\textwidth]{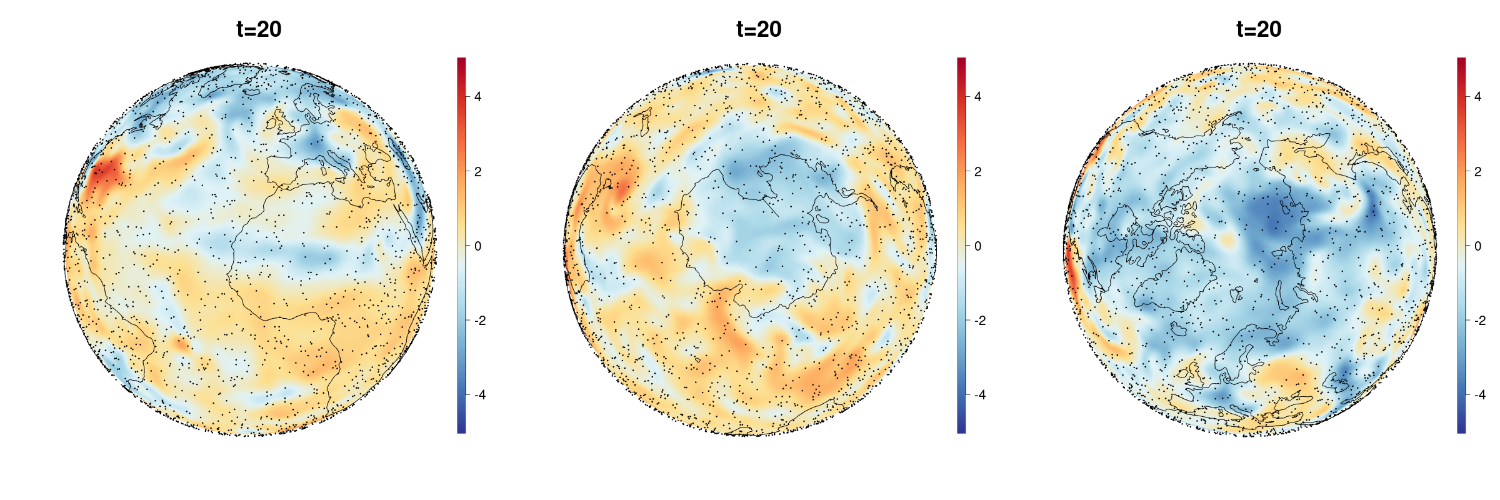}}
	\caption{\link{https://mike-pereira.github.io/PRED_STRF/krig-aod}{Spatio-temporal kriging of AOD on the globe (advection-diffusion model) represented from three different viewpoints on the surface at three different time steps. The black points locate the places where the observations are taken.}}\label{fig:krig_aod}
\end{figure}

We compare the predictions obtained from the advection-diffusion SPDE model with predictions obtained with a simpler model: a purely diffusive SPDE~\cite{lindgren2024diffusion}. Starting from the same observations, we therefore fit a model which does not account for advection, and therefore parametrized by a vector  $\bm \nu_{\text{D}}=(\kappa^2,c,\kappa^2_{in},\kappa_S^2,\tau,\sigma)^T$. We compare this model to the advection-diffusion model based on several criteria. 

First, the root mean squared error (RMSE) between the kriging predictors and the observed data, averaged over all spatial locations on the sphere, is computed at each time step. In addition, an overall spatio-temporal RMSE is evaluated. The results are reported in \Cref{fig:rmse}: subfigure (A) displays the spatial RMSE for the time steps used in the inference procedure (0 to 10), while subfigure (B) shows the spatial RMSE at both inference time steps (0 to 10) and subsequent time steps (11 to 20), where observations are used for prediction only.

In both settings, the advection–diffusion model outperforms the diffusion model, with an improvement of approximately $6\%$.

\Cref{fig:rmse}(B) further indicates that it is not necessary to perform inference using the full set of observations. When the model is calibrated using only the first 11 time steps, the predictive performance at later times remains consistent, and the relative behavior between the models is preserved. This suggests that the main conclusions are robust to a reduced inference window, and that reliable predictive performance can be achieved without exploiting the entire dataset for parameter estimation.

Finally, \Cref{fig:rmse_extrap} presents the spatial RMSE of extrapolated predictions over the six time steps following the inference window. In this setting, observations are assumed unavailable, so all predictions are fully out-of-sample. We additionally compare forecasts obtained from inference with partially observed fields (PO) to those derived from fully observed fields (FO) over time steps 0 to 10, as described in \Cref{sec:fo}. As expected, RMSE values are higher in this extrapolation setting due to the lack of observational data. Nevertheless, the qualitative behavior remains consistent: the PO advection–diffusion model systematically outperforms the PO diffusion model, with an improvement of
$3\%$. Moreover, the fully observed (FO) advection–diffusion model exhibits performance comparable to its partially observed (PO) counterpart, except for the first extrapolation step, where the FO model performs better due to the proximity to step 10, at which the field is observed at all locations. These results indicate that the advection–diffusion model, even when trained on incomplete fields, effectively captures transport dynamics and leverages them to produce reliable forecasts beyond the observation window.

\begin{figure}
    \centering
    \begin{subfigure}{0.48\textwidth}
        \includegraphics[width=\textwidth]{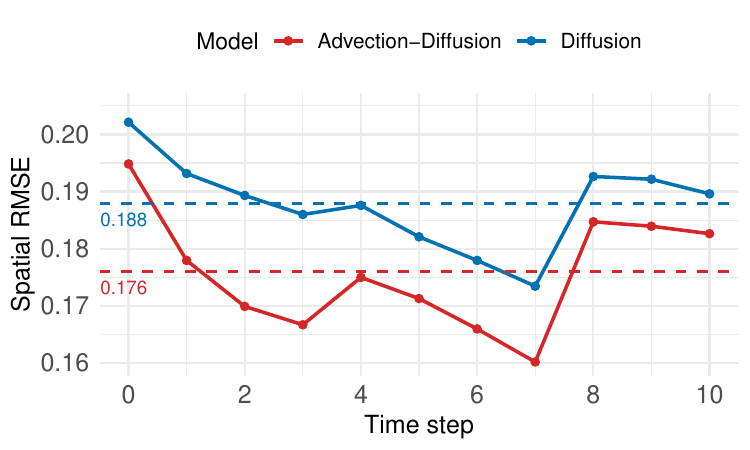}
        \caption{At inference time steps only}
    \end{subfigure}
    \begin{subfigure}{0.48\textwidth}
        \includegraphics[width=\textwidth]{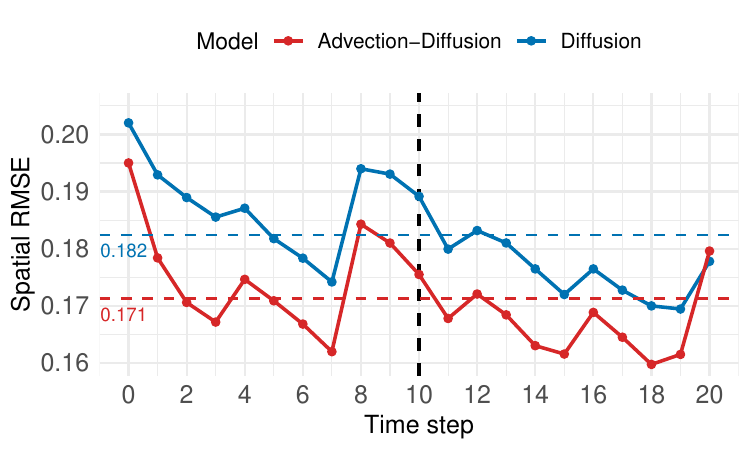}
        \caption{At all time steps}
    \end{subfigure}
    \caption{Spatial RMSE of the kriging predictors, at each time step, for the two prediction cases, with respective averages (dashed lines).}
    \label{fig:rmse}
\end{figure}

\begin{figure}
    \centering
    \includegraphics[width=0.6\textwidth]{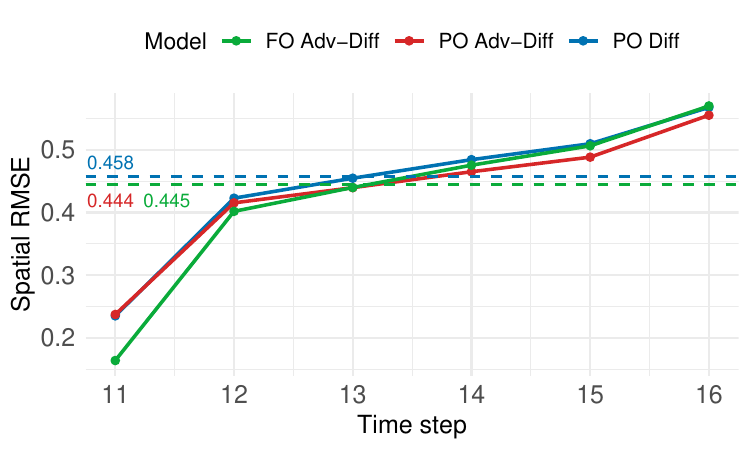}
    \caption{Spatial RMSE at each extrapolated time step, with respective averages (dashed lines).}
    \label{fig:rmse_extrap}
\end{figure}

Then, to assess the probabilistic prediction performance of the models, we rely on two scoring rules, the CRPS~\cite{matheson76} and the Variogram Score (VS)~\cite{scheuerer2015}, which are defined in Appendix~\ref{prob_score} and computed using 50 conditional simulations. The CRPS is used to evaluate the predictive performance of marginal distributions, whereas the VS is designed to assess the quality of bivariate dependence structures across pairs of locations. To compute these scores, we consider 32 random sets of 5000 (non-overlapping) spatial locations, which we keep fixed across the time steps. The weights in the VS score are chosen as the inverse of the total number of pairs of component combinations, and the order is set to $p=2$, so that the score is based on squared differences and places greater emphasis on larger discrepancies in spatial dependence. For each resulting set of spatio-temporal locations, we compute the average CRPS and the VS, and compare the obtained results. The result are shown in \Cref{fig:scores}.

The advection–diffusion model shows an improvement of $4\%$ in CRPS and of $6\%$ in VS over the diffusion model, which implies that the advection–diffusion model captures the spatial dependence and structure of the process more accurately.

\begin{figure}
    \centering
    \begin{subfigure}{0.48\textwidth}
        \includegraphics[width=\textwidth]{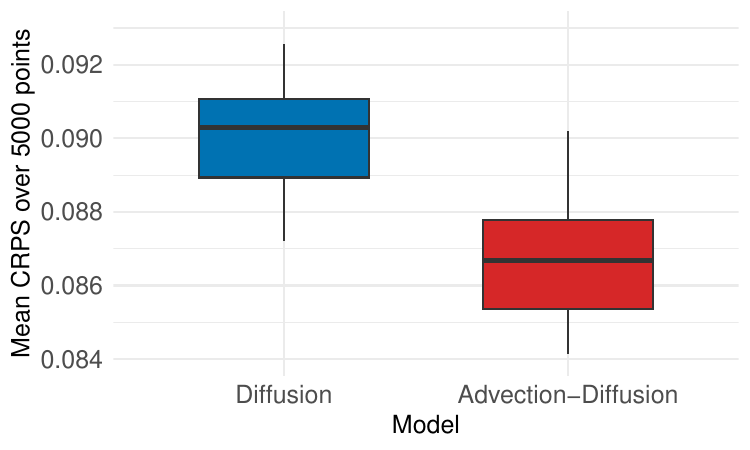}
        \caption{Average CRPS}
    \end{subfigure}
    \begin{subfigure}{0.48\textwidth}
        \includegraphics[width=\textwidth]{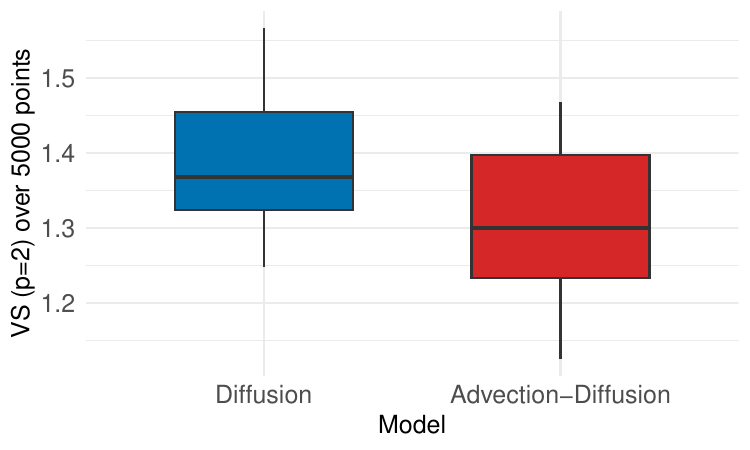}
        \caption{Variogram Score}
    \end{subfigure}
    \caption{Scores computed for the two models.}
    \label{fig:scores}
\end{figure}

As an illustration of how the two models differ, \Cref{fig:detail} presents their predictions at a given time step for both the advection–diffusion and diffusion models. The advection–diffusion model clearly captures the non-stationary channel-type behavior of AOD over central Antarctica. In contrast, the purely diffusion-based model fails to reproduce the underlying transport dynamics and performs worse in resolving the finer-scale structure of this feature.

\begin{figure}
    \centering
    \begin{subfigure}{0.25\textwidth}
        \includegraphics[width=\textwidth]{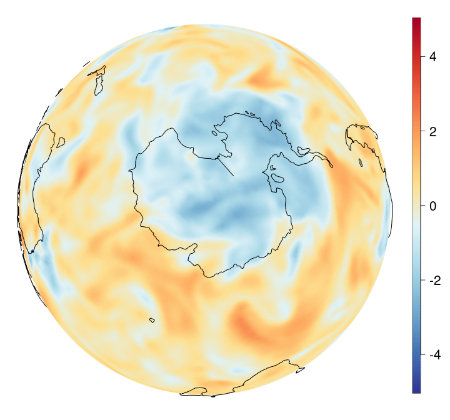}
        \caption{AOD data}
    \end{subfigure}
    \hspace{0.6cm}
    \begin{subfigure}{0.25\textwidth}
        \includegraphics[width=\textwidth]{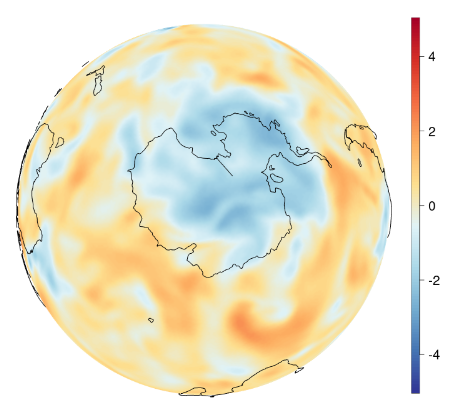}
        \caption{Advection-diffusion model}
    \end{subfigure}
    \hspace{0.6cm}
    \begin{subfigure}{0.25\textwidth}
        \includegraphics[width=\textwidth]{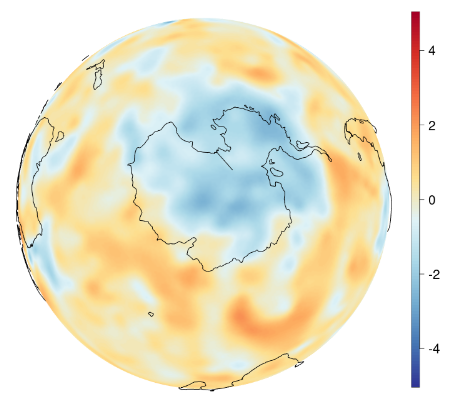}
        \caption{Diffusion model}
    \end{subfigure}
    \caption{AOD data and predictions with advection-diffusion and diffusion models at a given time step.}
    \label{fig:detail}
\end{figure}

This case study highlights the flexibility and scalability of SPDE-based Gaussian process models for analyzing high-resolution environmental data.

The reproducible code associated with this application is publicly available on GitHub at 
\url{https://github.com/mike-pereira/adv-diff-spde}.

\section{Conclusion and Discussion}
In this work, we introduced models for spatio-temporal GRFs on compact meshed surfaces, defined as solutions to advection–diffusion SPDEs. Our modeling framework easily accommodates non-stationary fields characterized by a temporally varying advection field. Our contribution follows recent developments in spatio-temporal statistics based on the SPDE formalism and offers the possibility to build and characterize increasingly precise and complex models. 

We derived numerically efficient algorithms for inference, simulation, and prediction by exploiting Galerkin discretization and matrix-free techniques. We showed how these algorithms can be implemented to achieve scalable computations even for very large meshes. The application to the aerosol dataset demonstrated that the model provides accurate predictions of a global environmental phenomenon governed by advection and diffusion processes. Moreover, it outperforms purely diffusive models in predictive performance, yielding improvements in both RMSE and probabilistic scores.

The SPDE model~\eqref{eq:spde} considered in this work assumes a polynomial diffusion term. That said, extending our framework to more general choices of $P$ is relatively straightforward. For instance, if $P$ is taken to be a function bounded away from zero, most of our analysis carries over with only minor modifications, thereby allowing for the inclusion of fractional operators. This generalization, however, introduces additional computational challenges. In particular, the matrices $\bm{\Gamma}^{(k)}$ would involve functions of matrices, which must be approximated numerically (e.g., using rational SPDE approximation methods \citep{bolin2020rational}, implemented in the \texttt{rSPDE} package \citep{bolin2023rspde}). Developing efficient implementations tailored to these specific operators is an interesting direction for future work.

Finally, future research could build on the proposed framework by further improving the inference approach, particularly by developing computationally feasible and efficient alternatives to gradient-free optimization. Further application-oriented developments could aim to integrate covariate-dependent fixed effects, enabling a deeper exploration of the aerosol physical process and enhancing the model's predictive skill.



\bibliographystyle{abbrvnat}
\bibliography{biblio}
\addcontentsline{toc}{section}{References}



\pagebreak

\appendix

\vspace{1em}
\begin{center}
	\Large\bfseries
	\textsc{Appendix}
\end{center}

\section{Mathematical tools}

\subsection{Functions of the Laplacian and colored noise}\label{sec:funcL}

Let $\lbrace \lambda_k\rbrace_{k\in\N}$ denote the set of eigenvalues of the Laplace–
Beltrami operator $\lb$ on $(\mathcal{M},g)$, and $\lbrace e_k\rbrace_{k\in\N}$ denote the associated eigenfunctions. In particular, $\lbrace e_k\rbrace_{k\in\N}$ form a basis of the space $L^2(\mathcal{M})$.
For $f : \R_+ \rightarrow \R$, we define $\mathcal{D}_f \subset L^2(\mathcal{M})$ as
\begin{equation*}
	\mathcal{D}_f = \lbrace \phi \in L^2(\mathcal{M}) : \sum_{j\in\N} f(\lambda_j)^2 \langle \phi, e_j\rangle^2 < \infty \rbrace.
\end{equation*}
Note that if $f$ is bounded, then $\mathcal{D}_f = L^2(\mathcal{M})$.
Then, the \textit{function of the Laplacian} $f(\lb)$ is the operator $f(\lb) : \mathcal{D}_f \rightarrow L^2(\mathcal{M})$ defined by 
\begin{equation*}
	f(\lb)\phi = \sum_{j\in\N} f(\lambda_j) \langle \phi, e_j\rangle e_j, \quad \phi\in L^2(\mathcal{M}).
\end{equation*}
Let then $f_S : \R_+ \rightarrow \R$ be a bounded function, and $\lbrace w_j\rbrace_{j\in\N}$ be a sequence of independent standard Gaussian variable. We call \textit{colored noise} the linear functional $f_S(-\Delta_{\mathcal{M}})\mathcal{W}_S$ defined by
\begin{equation*}
	f_S(-\Delta_{\mathcal{M}})\mathcal{W}_S : \phi \in H \mapsto \langle f_S(-\Delta_{\mathcal{M}})\mathcal{W}_S, \phi \rangle=\sum_{j\in\N} w_j f_S(\lambda_j) \langle \phi, e_j \rangle.
\end{equation*}
Note that for any $\phi \in H$, the series $\langle f_S(-\Delta_{\mathcal{M}})\mathcal{W}_S, \phi \rangle$ converges in quadratic mean since $\e[\langle f_S(-\Delta_{\mathcal{M}})\mathcal{W}_S, \phi \rangle]=0$, and by independence of the variables $w_j$, 
\begin{equation*}
	\e[\vert\langle f_S(-\Delta_{\mathcal{M}})\mathcal{W}_S, \phi \rangle\vert^2]
	=\sum_{j\in\N} \vert f_S(\lambda_j)\vert^2 \vert\langle \phi, e_j \rangle\vert^2
	\le (\sup_{\R_+} \vert f_S\vert)^2 \Vert \phi\Vert_S^2 < \infty.
\end{equation*}
Besides, using the same arguments, we have for any $\phi_1, \phi_2 \in H$,
\begin{equation*}
	\cov[\langle f_S(-\Delta_{\mathcal{M}})\mathcal{W}_S, \phi_1 \rangle, 
	\langle f_S(-\Delta_{\mathcal{M}})\mathcal{W}_S, \phi_2 \rangle] 
	=\sum_{j\in\N} \vert f_S(\lambda_j)\vert^2 \langle \phi_1, e_j \rangle \langle \phi_2, e_j \rangle
	=\langle f_S(-\Delta_{\mathcal{M}})\phi_1, f_S(-\Delta_{\mathcal{M}})\phi_2 \rangle,
\end{equation*}
Hence, in the case where $f_S(\lambda)=1$ for any $\lambda\ge 0$, $f_S(-\Delta_{\mathcal{M}})\mathcal{W}_S=\mathcal{W}_S$ corresponds to the definition of the spatial Gaussian white noise on $L^2(\mathcal{M})$. Also, whenever $f_S$ satisfies $f_S(\lambda)=\mathcal{O}_{\lambda\rightarrow\infty}(\lambda^{-\alpha})$ with $\alpha > d/4$, $f_S(-\Delta_{\mathcal{M}})\mathcal{W}_S$ can be identified with a square-integrable $H$-valued random variable, and decomposed as \cite[Proposition 2.7]{lang2023galerkin}:
\begin{equation*}
	f_S(-\Delta_{\mathcal{M}})\mathcal{W}_S=
	\sum_{j\in\N} w_j f_S(\lambda_j)  e_j.
\end{equation*}
For instance, if $f_S(\lambda)= \vert \kappa^2 +\lambda\vert^{-\alpha}$, then $\mathcal{Y}_S=f_S(-\Delta_{\mathcal{M}})\mathcal{W}_S$ can be seen as a solution of the Whittle-Matérn SPDE 
\begin{equation}
	(\kappa^2 \lb)^{\alpha}\mathcal{Y}_S=\mathcal{W}_S
\end{equation}
and can thus be seen as a Whittle-Matérn random field on $\mathcal{M}$. 

\subsection{Galerkin approximation}\label{sec:Galerkin}

For $N\in\N$, let $U_N=\spn{\varphi_k : 1\le k\le N}$ where $\varphi_1, \dots, \varphi_N \in H^1(\mathcal{M})$  are linearly independent functions. 	
The Galerkin approximation of $\lb$ over $U_N$ is the linear operator $-\Delta_N : U_N \rightarrow U_N$which maps any $\phi \in U_N$ to the element $-\Delta_N\phi\in U_N$ satisfying
\begin{equation*}
	\langle -\Delta_N\phi, \phi' \rangle = \langle \nabla\phi, \nabla \phi'\rangle, \quad \text{for any } \phi'\in U_N.
\end{equation*}
Similarly, for a smooth vector field $v$, we can define the Galerkin approximation of the operator $ f \in H^1(\mathcal{M}) \mapsto \dv(v f)$ as the linear operator  $\dv_N(v \,\cdot) : U_N \rightarrow U_N$ which maps any $\phi \in U_N$ to the element $\dv_N(v \phi)\in U_N$ satisfying
\begin{equation*}
	\langle \dv_N(v \phi), \phi' \rangle = \langle \dv(v \phi),  \phi'\rangle, \quad \text{for any } \phi'\in U_N.
\end{equation*}

As defined, $-\Delta_N$ is a symmetric endomorphism, and as such is diagonalizable. Let $\lbrace \lambda_k^{(N)}\rbrace_{1\le k\le N}$ denote its eigenvalues, and let $\lbrace e_k^{(N)}\rbrace_{1\le k\le N}$ be a set of associated eigenfunctions forming an orthonormal basis of $U_N$.

Let $\bm C, \bm R \in \R^{N\times N}$ be the matrices whose entries are respectively given by
\begin{equation}
	C_{ij}=\langle \varphi_i, \varphi_j\rangle, 
	\quad R_{ij}=\langle \nabla\varphi_i, \nabla\varphi_j\rangle, \quad 1\le i, j \le N.
\end{equation}
Note that, following \cite[Corollary 3.2]{lang2023galerkin}, $\lbrace \lambda_k^{(N)}\rbrace_{1\le k\le N}$ are also the eigenvalues  of the matrix $\widetilde{\bm R}=(\sqrt{\bm C})^{-1}\bm R(\sqrt{\bm C})^{-T}$. Besides, the map
\begin{equation}
	E : \bm v \in\R^N \mapsto \sum_{k=1}^N \bigg[(\sqrt{\bm C})^{-T}\bm v\bigg]_k \varphi_k \in U_N
	\label{eq:def_E}
\end{equation}
is an  isomorphism that maps the eigenvectors of $\widetilde{\bm R}$ to the eigenfunctions of $-\Delta_N$, and an isometry between $\R^N$ (equipped with the Euclidean metric) and $U_N$ (equipped with the metric induced by the inner product $\langle \cdot,\cdot\rangle$).

\subsection{Matrix functions}\label{sec:funcMat}

Let $\bm S\in\R^{N\times N}$ be a real symmetric matrix and let $f : \R \rightarrow \R$. In particular let us denote by $\lambda_1, \dots,\lambda_N$ the eigenvalues of $\bm S$ and let $\bm V \in\R^{N\times N}$ be an orthogonal matrix such that
\begin{equation}
	\bm S = \bm V \begin{smallmat}
		\lambda_1 & & \\
		& \ddots & \\
		& & \lambda_N
	\end{smallmat}\bm V^T.
	\label{eq:defV}
\end{equation}
Then, the matrix function $f(\bm S) \in\R^{N\times N}$ is the matrix defined by
\begin{equation*}
	f(\bm S) = \bm V \begin{smallmat}
		f(\lambda_1) & & \\
		& \ddots & \\
		& & f(\lambda_N)
	\end{smallmat}\bm V^T.
\end{equation*}
Note in particular that this definition is independent of the choice of matrix $\bm V$ in \eqref{eq:defV}, and that when $f$ is a polynomial, $f(\bm S)$ coincides with the usual notion of matrix polynomial.

\section{Analogy between SPDE formulations}\label{sec:analogy}

As defined, the forcing term $\mathcal{W}_T\otimes \mathcal{Y}_S$ in SPDE~\eqref{eq:spde} can be identified with a cylindrical Wiener process
$\lbrace\widetilde{\mathcal{W}}_t\rbrace_{t\in [0,T]}$ in $L^2(\mathcal{M})$  through 
\begin{equation*}
	\widetilde{\mathcal{W}}_t(\phi_S)=(\mathcal{W}_T\otimes \mathcal{Y}_S)(\ind_{[0,t]},\phi_S), \quad \phi_S\in L^2(\mathcal{M}),\; t\in [0,T],
\end{equation*}
where $\ind_{[0,t]}$ denotes the indicator function of the segment $[0,t]$ \cite{brehier2014short}. As such, we have (almost-surely) the following decomposition of $\widetilde{\mathcal{W}}_t$
\begin{equation}
	\widetilde{\mathcal{W}}_t = \sum_{j\in\N} f_S(\lambda_j) \beta_j(t) e_j, \quad t\in [0,T],
	\label{eq:cyl}
\end{equation}
where $\lbrace e_j\rbrace_{j\in\N}$ denotes an orthonormal basis of $L^2(\mathcal{M})$ composed eigenfunctions of the Laplace--Beltrami operator $\lb$, and $\lbrace \lambda_j\rbrace_{j\in\N}$ their associated eigenvalues.
This identification allows in turn to write
\begin{equation*}
	\mathcal{W}_T\otimes \mathcal{Y}_S(\phi_T,\phi_S) 
	= \langle \int_0^T \phi_T \dd\widetilde{\mathcal{W}}_t, \phi_S \rangle, \quad (\phi_T, \phi_S)\in L^2([0,T])\times L^2(\mathcal{M}),
\end{equation*}
where the integral term is given by
\begin{equation*}
	\int_0^T \phi_T \dd\widetilde{\mathcal{W}}_t = \sum_{j\in\N} f_S(\lambda_j) \bigg(\int_0^T \phi_T \dd\beta_j(t)\bigg) e_j.
\end{equation*}
Hence, we can interpret the forcing term $\mathcal{W}_T\otimes \mathcal{Y}_S$ as the (time) derivative of the cylindrical Wiener process $\lbrace\widetilde{\mathcal{W}}_t\rbrace_{t\in [0,T]}$. This analogy allows in particular to rewrite SPDE~\eqref{eq:spde} as in~\eqref{eq:spde-sde}.

The same analogy holds for the discretized forcing term $\mathcal{W}_T\otimes {Y}_S$ considered in the discretized SPDE~\eqref{eq:spde_discr}. It can be identified  with a cylindrical Wiener process
$\lbrace\widetilde{{W}}_t\rbrace_{t\in [0,T]}$ in $V_N^\ell$, which can be decomposed as~\eqref{eq:cyl} after replacing the eigenfunctions $\lbrace e_j\rbrace_{j\in\N}$ and eigenvalues $\lbrace \lambda_j\rbrace_{j\in\N}$ of $\lb$, by the eigenfunctions $\lbrace e_j^{(N)}\rbrace_{1\le j\le N}$ and eigenvalues $\lbrace \lambda_j^{(N)}\rbrace_{1\le j\le N}$ of $-\Delta_{N}$.

\section{Proofs of Section~\ref{sec:theory}}

\subsection{Proof of \Cref{prop:rec_ad}}\label{proof:rec_ad}

\begin{proof}
	We write for $0\le k<K$, $\bm \alpha^{(k)}=(\alpha_1^{(k)},\dots, \alpha_N^{(k)})^T\in\R^N$, $\bm y^{(k)}=(y_1^{(k)},\dots, y_N^{(k)})^T$, where 
	\begin{equation*}
		P(-\Delta_N)Z^{(k)}=\sum_{j=1}^N \alpha_j^{(k)}\psi_j^\ell, \quad Y^{(k)}=\sum_{j=1}^N y_j^{(k)}\psi_j^\ell.
	\end{equation*}
	Besides, let $\bm\xi^{(k)}$ be the vector defined by $\bm \xi^{(k)} = \big(\langle Z^{(k)}, e_1^{(N)}\rangle, \dots, \langle Z^{(k)}, e_N^{(N)}\rangle\big)^T$. 
	
	Firstly, note that, following~\citep[Theorem 3.4]{lang2023galerkin}, we can take
	\begin{equation*}
		\bm y^{(k)}=(\sqrt{\bm C})^{-T}f_S(\widetilde{\bm R})\bm w ^{(k)}, \quad k\in\N,
	\end{equation*}
	where $\lbrace\bm w ^{(k)}\rbrace_{k\in\N}$ is a sequence of independent centered Gaussian vectors with covariance matrix $\bm I$.

	Then, by testing~\eqref{eq:euler_discr} against $\psi_i$ (for $i\in\lbrace 1,\dots,N\rbrace$), and injecting~\eqref{eq:dd}, we get the following linear system of equations
	\begin{equation}
		\bm C \bm z^{(k+1)} - \bm C \bm z^{(k)} +\frac{\delta t}{c}\bigg(\bm C \bm \alpha^{(k+1)} + \bm B^{(k)}\bm z^{(k+1)}\bigg)
		=\frac{\delta t}{c}\bm C\bm m+\tau\sqrt{\frac{\delta t}{c}}\bm C \bm y ^{(k+1)}
		\label{eq:zk_scheme}
	\end{equation}
	On the one hand, using the map $E : E : \bm v \in\R^N \mapsto \sum_{k=1}^N \big[(\sqrt{\bm C})^{-T}\bm v\big]_k \psi_k \in V_N $, we have
	\begin{equation}
		Z^{(k+1)}
		=E\bigg((\sqrt{\bm C})^T\bm z^{(k+1)}\bigg) \quad \text{and}\quad P(-\Delta_N)Z^{(k+1)}
		=E\bigg((\sqrt{\bm C})^T\bm \alpha^{(k+1)}\bigg)
		\label{eq:eq_zk}
	\end{equation}
	On the other hand, following the definition of the basis $\lbrace e_j^{(N)}\rbrace_{1\le j\le N}$,
	\begin{equation*}
		\begin{aligned}
			Z^{(k+1)}
			&
			=\sum_{j=1}^N  \langle Z^{(k+1)}, e_j^{(N)}\rangle E(\bm v_j)
			=E\bigg(\sum_{j=1}^N \langle Z^{(k+1)}, e_j^{(N)}\rangle \bm v_j\bigg)
			=E\big(\bm V\bm\xi^{(k+1)}\big)
		\end{aligned}
	\end{equation*}
	Hence, since $E$ is invertible, we have $(\sqrt{\bm C})^T\bm z^{(k+1)} = \bm V\bm\xi^{(k+1)}$.
	
	Similarly, by definition of $P(-\Delta_N)$, we have
	\begin{equation*}
		\begin{aligned}
			P(-\Delta_N)Z^{(k+1)}
			&=\sum_{j=1}^N P(\lambda_j^{(N)})\langle Z^{(k+1)}, e_j^{(N)}\rangle e_j^{(N)}
			=E\big(\bm VP(\bm \Lambda^{(N)})\bm\xi^{(k+1)}\big)
			=E\big(P(\widetilde{\bm R})\bm V\bm\xi^{(k+1)}\big)
		\end{aligned}
	\end{equation*}
	Therefore, we have $P(-\Delta_N)Z^{(k+1)}=E\big(P(\widetilde{\bm R})(\sqrt{\bm C})^T\bm z^{(k+1)}\big)$, and using~\eqref{eq:eq_zk} and the fact that $E$ is invertible, we can deduce that  $\bm \alpha^{(k+1)}=
	(\sqrt{\bm C})^{-T}P(\widetilde{\bm R})(\sqrt{\bm C})^T\bm z^{(k+1)}$. 
	
	In conclusion, we can now rewrite~\eqref{eq:zk_scheme} as
	\begin{equation*}
		\bigg(\bm C  +\frac{\delta t}{c} (\sqrt{\bm C})P(\widetilde{\bm R})(\sqrt{\bm C})^T + \frac{\delta t}{c}\bm B^{(k)}\bigg)\bm z^{(k+1)}
		=\bm C \bm z^{(k)}+\frac{\delta t}{c}\bm C\bm m+\tau\sqrt{\frac{\delta t}{c}}\bm C \bm y ^{(k+1)}.
	\end{equation*}
	By then multiplying both sides of the equality by $(\sqrt{\bm C})^{-1}$ and  introducing $\bm x^{(k)}=(\sqrt{\bm C})^T \bm z^{(k)}$, we retrieve the recursion~\eqref{eq:rec_ad}. 
\end{proof}

\subsection{Proof of \Cref{prop:mean_cov_ad}}\label{proof:mean_cov_ad}

\begin{proof}
		For ease of notation, let us write for $k\ge 0$, $\bm \Gamma^{-k}=(\bm \Gamma^{-1})^k$. On the one hand, by taking the expectation on both sides of the recursion~\eqref{eq:rec_ad}, we find $\bm\mu^{(0)}=\e[\bm z^{(0)}]=(\sqrt{\bm C})^{-T} \e[\bm x^{(0)}]$ and for $k\ge 0$,
	\begin{equation*}
		\left\lbrace\begin{aligned}
			& \bm\Gamma\e[\bm x^{(k+1)}]
			=\e[\bm x^{(k)}]+\frac{\delta t}{c}\widetilde{\bm m},\\
			&\bm\mu^{(k+1)}=\e[\bm z^{(k+1)}]=(\sqrt{\bm C})^{-T} \e[\bm x^{(k+1)}],
		\end{aligned}\right.
	\end{equation*}
	where we take $\widetilde{\bm m}=(\sqrt{\bm C})^T \bm m$. Hence, we have for any $k\ge 0$,
	\begin{align*}
		\e[\bm x^{(k)}]-\frac{\delta t}{c}(\bm\Gamma-\bm I)^{-1}\widetilde{\bm m}
		&=\bm\Gamma\e[\bm x^{(k+1)}]-\frac{\delta t}{c}\widetilde{\bm m}-\frac{\delta t}{c}(\bm\Gamma-\bm I)^{-1}\widetilde{\bm m}
		=\bm\Gamma\e[\bm x^{(k+1)}]-\frac{\delta t}{c}(\bm I+(\bm\Gamma-\bm I)^{-1})\widetilde{\bm m}\\
		&=\bm\Gamma\e[\bm x^{(k+1)}]-\frac{\delta t}{c}((\bm\Gamma-\bm I)+\bm I)(\bm\Gamma-\bm I)^{-1}\widetilde{\bm m}
		=\bm\Gamma\bigg(\e[\bm x^{(k+1)}]-\frac{\delta t}{c}(\bm\Gamma-\bm I)^{-1}\widetilde{\bm m}\bigg).
	\end{align*}
	By iterating this last equation, we obtain
	\begin{equation*}
		\e[\bm x^{(0)}]-\frac{\delta t}{c}(\bm\Gamma-\bm I)^{-1}\widetilde{\bm m}
		=\bm\Gamma^{k}\bigg(\e[\bm x^{(k)}]-\frac{\delta t}{c}(\bm\Gamma-\bm I)^{-1}\widetilde{\bm m}\bigg).
	\end{equation*}
	Finally, by multiplying both side of this equality by $(\sqrt{\bm C})^{-T}$, we get
	\begin{align*}
		\bm\mu^{(0)}-\frac{\delta t}{c}(\sqrt{\bm C})^{-T}(\bm\Gamma-\bm I)^{-1}\widetilde{\bm m}
		=\bm\mu^{(0)}-\bm\mu 
		&=(\sqrt{\bm C})^{-T}\bm\Gamma^{k}(\sqrt{\bm C})^{T}\bigg(\bm \mu^{(k)}-\frac{\delta t}{c}(\sqrt{\bm C})^{-T}(\bm\Gamma-\bm I)^{-1}\widetilde{\bm m}\bigg)\\
		&=(\sqrt{\bm C})^{-T}\bm\Gamma^{k}(\sqrt{\bm C})^{T}\bigg(\bm \mu^{(k)}-\bm \mu\bigg),
	\end{align*}
	which in turn gives~\eqref{eq:exp_z}.
	
	On the other hand, by taking the covariance on both sides of the recursion~\eqref{eq:rec_ad}, we have, for any $k\ge 0$,
	\begin{equation}~\label{eq:rec_cov}
		\left\lbrace\begin{aligned}
			& \bm\Gamma\var[\bm x^{(k+1)}]\bm\Gamma^T
			=\var[\bm x^{(k)}]+f^2_{\delta t}(\widetilde{\bm R}),\\
			&\bm\Sigma^{(k+1)}=(\sqrt{\bm C})^{-T} \var[\bm x^{(k+1)}](\sqrt{\bm C})^{-1},
		\end{aligned}\right.
	\end{equation}
Let us show by induction that for any $k\ge 0$,
	\begin{equation}\label{eq:induc}
	\var[\bm x^{(k)}]= \bm\Gamma^{-k}\var[\bm x^{(0)}](\bm\Gamma^{-k})^T +\sum_{i=1}^{k} \bm\Gamma^{-i}{f}_{\delta t}^2(\widetilde{\bm R})(\bm\Gamma^{-i})^T.
\end{equation}
Formula~\eqref{eq:induc} clearly holds for $k=0$. Let us now assume that it holds for some fixed index $k\in\N_0$. Then, following~\eqref{eq:rec_cov}, we have
	\begin{align*}
	\var[\bm x^{(k+1)}]
	&= \bm\Gamma^{-1}(\var[\bm x^{(k)}]+f^2_{\delta t}(\widetilde{\bm R})) \bm\Gamma^{-T}\\
	&= \bm\Gamma^{-(k+1)}\var[\bm x^{(0)}](\bm\Gamma^{-(k+1)})^T +\sum_{i=1}^{k} \bm\Gamma^{-(i+1)}{f}_{\delta t}^2(\widetilde{\bm R})(\bm\Gamma^{-(i+1)})^T
	+\bm\Gamma^{-1}f^2_{\delta t}(\widetilde{\bm R})\bm\Gamma^{-T}\\
	&= \bm\Gamma^{-(k+1)}\var[\bm x^{(0)}](\bm\Gamma^{-(k+1)})^T +\sum_{i=1}^{k+1} \bm\Gamma^{-i}{f}_{\delta t}^2(\widetilde{\bm R})(\bm\Gamma^{-i})^T
\end{align*}
Hence, Formula~\eqref{eq:induc} also holds for the index $(k+1)$. In conclusion, by induction, Formula~\eqref{eq:induc} holds for any $k\ge 0$. We then retrieve~\eqref{eq:cov_z} as a consequence of~\eqref{eq:rec_cov}.
\end{proof}

\subsection{Proof of \Cref{prop:Qeuler}}\label{proof:Qeuler}

\begin{proof}
	Consider the vectors $\lbrace\bm w^{(k)}, \bm x^{(k)}\rbrace_{0\le k\le K}$ defined in \Cref{prop:rec_ad}, and define the vectors $\bm W =((\bm w^{(0)})^T, \dots, (\bm w^{(K)})^T)^T\in\R^{(K+1)N}$ and $\bm X =((\bm x^{(0)})^T, \dots, (\bm x^{(K)})^T)^T \in \R^{(K+1)N}$. In particular,  $\bm W$ is a centered Gaussian vector with precision matrix $\bm Q_{\bm W}=\bm I$. Besides, note that by definition of $Z(0,\cdot)$ we can write $\bm z^{(0)}={\bm m}^{(0)}+(\sqrt{\bm C})^{-T}f_0(\widetilde{\bm R})\bm w^{(0)}$ . On the one hand, we can rewrite the first equality of~\eqref{eq:rec_ad} in matrix form as
	\begin{equation}\label{eq:matform}
		\bm L\left(\lbrace\bm \Gamma^{(k)}\rbrace_{k=0}^{K-1}\right)\bm X=\bm D\big((\sqrt{\bm C})^{T}\big){\bm M}_{\delta t}+\bm D\big(f_0(\widetilde{\bm R}),\, {f}_{\delta t}(\widetilde{\bm R})\big)\bm W.
	\end{equation}
	Taking the expectation on both sides of~\eqref{eq:matform}, we then get $\e[\bm X]=\bm L\left(\lbrace\bm \Gamma^{(k)}\rbrace_{k=0}^{K-1}\right)^{-1}\bm D\big((\sqrt{\bm C})^{T}\big){\bm M}_{\delta t}$.
	
	On the one hand, since by definition of $\bm{x}^{(k)}$ ($k\ge 0$) in \Cref{prop:rec_ad}, we have $\bm X =  \bm D\big((\sqrt{\bm C})^{T})\bm Z$, we can conclude that 
	\begin{equation*}
		\bm \mu_{\bm Z}=\bm D\big((\sqrt{\bm C})^{-T})\bm L\left(\lbrace\bm \Gamma^{(k)}\rbrace_{k=0}^{K-1}\right)^{-1}\bm D\big((\sqrt{\bm C})^{T}\big){\bm M}_{\delta t}.
	\end{equation*}
	On the other hand, we can also write
	\begin{equation}\label{eq:matformb}
		\bm L\left(\lbrace\bm \Gamma^{(k)}\rbrace_{k=0}^{K-1}\right)(\bm X-\e[\bm X])=\bm D\big(f_0(\widetilde{\bm R}),\, {f}_{\delta t}(\widetilde{\bm R})\big)\bm W.
	\end{equation}
	Taking the covariance of both side of this equality, we get 
	\begin{equation*}
		\bm L\left(\lbrace\bm \Gamma^{(k)}\rbrace_{k=0}^{K-1}\right) \bm Q_{\bm X}^{-1}
		\bm L\left(\lbrace\bm \Gamma^{(k)}\rbrace_{k=0}^{K-1}\right)^T
		= \bm D\big(f_0^2(\widetilde{\bm R}),\, {f}_{\delta t}^2(\widetilde{\bm R})\big),
	\end{equation*}
	where we used the fact that $\bm D\big(f_0(\widetilde{\bm R}),\, {f}_{\delta t}(\widetilde{\bm R})\big)$ is block diagonal and that $\bm Q_{\bm W}=\bm I$.
	By then inverting both sides, we obtain the following relation for the precision matrix $\bm Q_{\bm X}$ of $\bm X$:
	\begin{equation*}
		\bm Q_{\bm X} =  \bm L\left(\lbrace\bm \Gamma^{(k)}\rbrace_{k=0}^{K-1}\right)^T~\bm D\left(f_0^{-2}(\widetilde{\bm R}),\, {f}_{\delta t}^{-2}(\widetilde{\bm R})\right)~\bm L\left(\lbrace\bm \Gamma^{(k)}\rbrace_{k=0}^{K-1}\right).
	\end{equation*}
	Finally, recall that $\bm X =  \bm D\big((\sqrt{\bm C})^{T})\bm Z$. By once again taking the covariance and then inverting both sides of this equality, we retrieve~\eqref{eq:defQ}.
\end{proof}

\section{Stability of the recursion}

\begin{prop}\label{prop:stab}
	Assume that the finite element matrices $\bm C$, $\bm R$ and $\bm B$ are defined according to~\eqref{eq:coef} and let ${\bm G}=\frac{1}{2}(\bm B+\bm B^T)$. Then we have, 
	\begin{equation*}
		{G}_{ij}= \langle\frac{1}{2}\dv(\gamma)\psi_i^\ell,\psi_j^\ell\rangle, \quad 1\le i,j\le n.
	\end{equation*}
	In particular, we have $C_P+\displaystyle\lambda_{\min}(\tilde{\bm G})\ge \inf_{\mathcal{M}} \bigg(C_P +  \frac{1}{2}\dv(\gamma)\bigg)$ meaning that the recursion in \Cref{prop:rec_ad} is stable whenever $$\inf_{\mathcal{M}} \bigg(C_P +  \frac{1}{2}\dv(\gamma)\bigg)>0.$$
\end{prop}

\begin{proof}
	Let $ 1\le i,j\le n$. On the one hand, we have using the Leibnitz rule,
	\begin{align*}
		\langle\psi_j^\ell,\dv(\gamma\psi_i^\ell)\rangle
		=\langle\psi_j^\ell,\dv(\gamma)\psi_i^\ell\rangle+\langle\psi_j^\ell,g(\gamma,\nabla\psi_i^\ell)\rangle
	\end{align*}
	On the other hand, the integration by parts formula gives
		\begin{align*}
		\langle\psi_i^\ell,\dv(\gamma\psi_j^\ell)\rangle
		=-\langle\nabla\psi_i^\ell,\gamma\psi_j^\ell\rangle
		=-\langle g(\nabla\psi_i^\ell,\gamma),\psi_j^\ell\rangle
	\end{align*}
	Hence, we can write
	\begin{align*}
	2 {G}_{ij}
	=\langle\psi_i^\ell,\dv(\gamma\psi_j^\ell)\rangle+\langle\psi_j^\ell,\dv(\gamma\psi_i^\ell)\rangle
	=-\langle g(\nabla\psi_i^\ell,\gamma),\psi_j^\ell\rangle+\langle\psi_j^\ell,\dv(\gamma)\psi_i^\ell\rangle+\langle\psi_j^\ell,g(\gamma,\nabla\psi_i^\ell)\rangle
\end{align*}
which gives
	\begin{align*}
	{G}_{ij}
	=\frac{1}{2}\langle\psi_j^\ell,\dv(\gamma)\psi_i^\ell\rangle
\end{align*}

Following the proof of \Cref{prop:mean_cov_inf}, the recursion is stable if $(P(\widetilde{\bm R}) + \widetilde{\bm G})$ is positive definite. Let then $\bm x\in\R^n$ be fixed but arbitrary. Let us take $\widetilde{\bm x}=(\sqrt{\bm C})^{-T}\bm x$ and $\varphi=\sum_{i=1}^n\widetilde{x}_i\psi_i^\ell \in V_{N}$. Then,`
\begin{equation*}
	\bm x^T  \left(P(\widetilde{\bm R}) + \widetilde{\bm G}\right) \bm x \geq \bm x^T(C_P\bm I +\widetilde{\bm G})\bm x=\widetilde{\bm x}^T(C_P\bm C +{\bm G})\widetilde{\bm x}=C_P\widetilde{\bm x}^T\bm C\widetilde{\bm x}+\widetilde{\bm x}^T{\bm G}\widetilde{\bm x}
\end{equation*}
where in particular $\widetilde{\bm x}^T\bm C\widetilde{\bm x} =\sum_{1\le i,j\le n}\widetilde{x}_i\langle\psi_i,\psi_j\rangle\widetilde{x}_j=\langle\varphi,\varphi\rangle$ and $\widetilde{\bm x}^T\bm G\widetilde{\bm x} =\sum_{1\le i,j\le n}\widetilde{x}_i\langle \frac{1}{2}\dv(\gamma)\psi_i,\psi_j\rangle\widetilde{x}_j=\langle \frac{1}{2}\dv(\gamma)\varphi,\varphi\rangle$. Hence, we have
\begin{equation*}
	\bm x^T  \left(P(\widetilde{\bm R}) + \widetilde{\bm G}\right) \bm x \geq \langle \bigg(C_P+\frac{1}{2}\dv(\gamma)\bigg)\varphi,\varphi\rangle
	\ge \inf_{\mathcal{M}}\bigg(C_P+\frac{1}{2}\dv(\gamma)\bigg)\langle \varphi,\varphi\rangle=\inf_{\mathcal{M}}\bigg(C_P+\frac{1}{2}\dv(\gamma)\bigg)\widetilde{\bm x}^T\bm C\widetilde{\bm x}
\end{equation*}
Finally, note that by definition of $\widetilde{\bm x}$, we have $\widetilde{\bm x}^T\bm C\widetilde{\bm x}=\bm x^T\bm x$, thus giving that, for any  $\bm x\in\R^n$,
\begin{equation*}
	\bm x^T  \left(P(\widetilde{\bm R}) + \widetilde{\bm G}\right) \bm x \geq  \inf_{\mathcal{M}}\bigg(C_P+\frac{1}{2}\dv(\gamma)\bigg)\Vert\bm x\Vert^2,
\end{equation*}
which proves that $(P(\widetilde{\bm R}) + \widetilde{\bm G})$ is indeed positive definite.

\end{proof}

\section{Algorithms}

We expose in this section a few algorithms that are necessary to perform matrix-free inference and prediction.

\begin{algorithm}
	\caption{Matrix-vector product by $\bm Q_{\bm Z}$}\label{alg:mvprodQ}
	\begin{algorithmic}[1]
		
		\KwDepend{Matrices $(\sqrt{\bm C}), \widetilde{\bm R}, \lbrace\bm\Gamma^{(k)}\rbrace_{k=0}^{K-1}$ in~\eqref{eq:defQ}.}
		\KwIn{Vector $\bm X =((\bm x^{(0)})^T, \dots, (\bm x^{(K)})^T)^T\in\R^{(K+1)N}$.}
		\KwOut{Vector $\bm Y=((\bm y^{(0)})^T, \dots, (\bm y^{(K)})^T)^T=\bm Q_{\bm Z} \bm X$.}
		
		\vspace{1ex}
		\hrule
		\vspace{1ex}

		\For{$k=0$ \textbf{to} $K$}
		\State  Initialize $\bm y^{(k)}=(\sqrt{\bm C})^T\bm x^{(k)}$
		\State  Initialize $\bm z^{(k)}=\bm 0$
		\EndFor
		
		\vspace{1ex} 
		
		\State Set $\bm z^{(0)}\leftarrow \bm y^{(0)}$
		\For{$k=1$ \textbf{to} $K$}
		\State  Set $\bm z^{(k)}\leftarrow \bm \Gamma^{(k-1)} \bm y^{(k)}-\bm y^{(k-1)}$
		\EndFor
		
		\vspace{1ex}
		
		\State  Set $\bm z^{(0)} \leftarrow f_0^{-2}(\widetilde{\bm R})\bm z^{(0)}$
		\For{$k=1$ \textbf{to} $K$}
		\State Set $\bm z^{(k)}\leftarrow {f}_{\delta t}^{-2}(\widetilde{\bm R})\bm z^{(k)}$
		\EndFor
		
		\vspace{1ex}

		\State Set $\bm y^{(K)} \leftarrow  (\bm \Gamma^{(K-1)})^{T}\bm z^{(K)}$
		\For{$k=K-1$ \textbf{to} $1$}
		\State Set  $\bm y^{(k)}\leftarrow (\bm \Gamma^{(k-1)})^{T}\bm z^{(k)}-\bm z^{(k+1)}$
		\EndFor
		\State Set  $\bm y^{(0)} \leftarrow  \bm z^{(0)}-\bm z^{(1)}$
		
		\vspace{1ex}
		
		\For{$k=0$ \textbf{to} $K$}
		\State  Set $\bm y^{(k)}\leftarrow (\sqrt{\bm C}) \bm y^{(k)}$
		\EndFor
		
		\vspace{1ex}
		
		\KwRet{$\bm Y=((\bm y^{(0)})^T, \dots, (\bm y^{(K)})^T)^T$.}
	\end{algorithmic}
\end{algorithm}

\begin{algorithm}
	\caption{Solve a linear system defined by $\bm L(\lbrace\bm\Theta^{(k)}\rbrace_{k=0}^{K-1})$}\label{alg:linsolveL}
	\begin{algorithmic}[1]
		
		\KwDepend{Matrices $\lbrace\bm\Theta^{(k)}\rbrace_{k=0}^{K-1}$ in~\eqref{eq:defL}.}
		\KwIn{Vector $\bm X =((\bm x^{(0)})^T, \dots, (\bm x^{(K)})^T)^T\in\R^{(K+1)N}$.}
		\KwOut{Vector $\bm Y=((\bm y^{(0)})^T, \dots, (\bm y^{(K)})^T)^T=\bm L(\lbrace\bm\Theta^{(k)}\rbrace_{k=0}^{K-1})^{-1} \bm X$.}
		
		\vspace{1ex}
		\hrule
		\vspace{1ex}
		
		\For{$k=0$ \textbf{to} $K$}
		\State  Initialize  $\bm y^{(k)}=\bm 0$
		\EndFor
		
		\vspace{1ex} 
		
		\State  Set $\bm y^{(0)} \leftarrow \bm x^{(0)}$
		\For{$k=1$ \textbf{to} $K$}
		\State  Set $\bm y^{(k)}\leftarrow (\bm \Theta^{(k-1)})^{-1}\big(\bm y^{(k-1)}+\bm x^{(k)}\big)$
		\EndFor
		
		\vspace{1ex}
		
		\KwRet{$\bm Y=((\bm y^{(0)})^T, \dots, (\bm y^{(K)})^T)^T$.}
	\end{algorithmic}
\end{algorithm}

\begin{algorithm}
	\caption{Solve a linear system defined by $\bm L(\lbrace\bm\Theta^{(k)}\rbrace_{k=0}^{K-1})^T$}\label{alg:linsolveLT}
	\begin{algorithmic}[1]
		
		\KwDepend{Matrices $\lbrace\bm\Theta^{(k)}\rbrace_{k=0}^{K-1}$ in~\eqref{eq:defL}.}
		\KwIn{Vector $\bm X =((\bm x^{(0)})^T, \dots, (\bm x^{(K)})^T)^T\in\R^{(K+1)N}$.}
		\KwOut{Vector $\bm Y=((\bm y^{(0)})^T, \dots, (\bm y^{(K)})^T)^T=\bm L(\lbrace\bm\Theta^{(k)}\rbrace_{k=0}^{K-1})^{-T} \bm X$.}
		
		\vspace{1ex}
		\hrule
		\vspace{1ex}
		
		\For{$k=0$ \textbf{to} $K$}
		\State  Initialize  $\bm y^{(k)}=\bm 0$
		\EndFor
		
		\vspace{1ex}
		
		\State  Set $\bm y^{(K)} \leftarrow (\bm \Theta^{(K-1)})^{-T}\bm x^{(K)}$
		\For{$k=K-1$ \textbf{to} $1$}
		\State  Set $\bm y^{(k)}= (\bm \Theta^{(k-1)})^{-T}\big(\bm y^{(k+1)}+\bm x^{(k)}\big)$
		\EndFor
		\State  Set $\bm y^{(0)}\leftarrow  \bm y^{(1)}+\bm x^{(0)}$
		
		\vspace{1ex} 
		
		\KwRet{$\bm Y=((\bm y^{(0)})^T, \dots, (\bm y^{(K)})^T)^T$.}
	\end{algorithmic}
\end{algorithm}

\begin{algorithm}
	\caption{Solve a linear system defined by $\bm Q_{\bm Z}$}\label{alg:linsolveQ}
	\begin{algorithmic}[1]
		
		\KwDepend{Matrices $(\sqrt{\bm C}), \widetilde{\bm R}, \lbrace\bm\Gamma^{(k)}\rbrace_{k=0}^{K-1}$ in~\eqref{eq:defQ}.}
		\KwIn{Vector $\bm X =((\bm x^{(0)})^T, \dots, (\bm x^{(K)})^T)^T\in\R^{(K+1)N}$.}
		\KwOut{Vector $\bm Y=((\bm y^{(0)})^T, \dots, (\bm y^{(K)})^T)^T=\bm Q_{\bm Z}^{-1} \bm X$.}
		
		\vspace{1ex}
		\hrule
		\vspace{1ex}
		
		\For{$k=0$ \textbf{to} $K$}
		\State  Initialize $\bm y^{(k)}=(\sqrt{\bm C})^{-1}\bm x^{(k)}$
		\EndFor
		
		\vspace{1ex}
		
		\State Set $\bm Y \leftarrow \bm L(\lbrace\bm\Gamma^{(k)}\rbrace_{k=0}^{K-1})^{-T} \bm Y$ using \Cref{alg:linsolveLT}
		
		\vspace{1ex}
		
		\State  Set $\bm y^{(0)} \leftarrow f_0^{2}(\widetilde{\bm R})\bm y^{(0)}$
		\For{$k=1$ \textbf{to} $K$}
		\State  Set $\bm y^{(k)}\leftarrow {f}_{\delta t}^{2}(\widetilde{\bm R})\bm y^{(k)}$
		\EndFor
		
		\vspace{1ex}
		
		\State Set $\bm Y\leftarrow \bm L(\lbrace\bm\Gamma^{(k)}\rbrace_{k=0}^{K-1})^{-1} \bm Y$ using \Cref{alg:linsolveL}
		
		\vspace{1ex}
		
		\For{$k=0$ \textbf{to} $K$}
		\State  Set $\bm y^{(k)}\leftarrow(\sqrt{\bm C})^{-T}\bm y^{(k)}$
		\EndFor
		
		\vspace{1ex}
		
		\KwRet{$\bm Y=((\bm y^{(0)})^T, \dots, (\bm y^{(K)})^T)^T$.}
	\end{algorithmic}
\end{algorithm}

\begin{algorithm}
    \caption{Symmetric Block Gauss-Seidel preconditioner}\label{alg:blockGSPrecond}
    \begin{algorithmic}[1]
        
        \KwDepend{A symmetric block tri-diagonal matrix $\mathbf{Q} = \mathbf{L} + \mathbf{D} + \mathbf{L}^T$ where $\mathbf{D} = \text{diag}(\mathbf{D}_0, \dots, \mathbf{D}_K)$ and $\mathbf{L}$ contains sub-diagonal blocks $\mathbf{M}_k$ for $k \in \{1, \dots, K\}$.}
        \KwIn{Vector $\bm X =((\bm x^{(0)})^T, \dots,  (\bm x^{(K)})^T)^T \in \mathbb{R}^{(K+1)N}$.}
        \KwOut{Vector $\bm Y=((\bm y^{(0)})^T, \dots, (\bm y^{(K)})^T)^T=(\mathbf{L}+\mathbf{D})^{-1}\mathbf{D}(\mathbf{L}^T+\mathbf{D})^{-1} \bm X$.}

        \vspace{1ex}
        \hrule
        \vspace{1ex}
        
        \State \textbf{Step 0: Initialization} \Comment{Once for all the iterations of the GMRES algorithm}
        \For{$k=0$ \textbf{to} $K$}
            \State Compute the Cholesky decomposition of $\mathbf{D}_k$
            \State $\mathbf{D}_k = \mathbf{G}_k \mathbf{G}_k^T$ where $\mathbf{G}_k$ is lower triangular
        \EndFor

        \vspace{1ex}
        \hrule
        \vspace{1ex}
        
        \State \textbf{Step 1: Forward sweep}
        \State Solve $\mathbf{G}_0 \mathbf{G}_0^T \bm y^{(0)} = \bm x^{(0)}$
        \For{$k=1$ \textbf{to} $K$}
            \State $\bm r^{(k)} \leftarrow \bm x^{(k)} - \mathbf{M}_{k}^T \bm y^{(k-1)}$
            \State Solve $\mathbf{G}_k \mathbf{G}_k^T \bm y^{(k)} = \bm r^{(k)}$
        \EndFor

        \vspace{1ex}

        \State \textbf{Step 2: Scaling}
        \For{$k=0$ \textbf{to} $K-1$}
            \State $\bm y^{(k)} \leftarrow \mathbf{D}_k \bm y^{(k)}$
        \EndFor

        \vspace{1ex}

        \State \textbf{Step 3: Backward sweep}
        \For{$k=K-1$ \textbf{down to} $0$}
            \State $\bm r^{(k)} \leftarrow \bm y^{(k)} - \mathbf{M}_{k+1} \bm y^{(k+1)}$
            \State Solve $\mathbf{G}_k \mathbf{G}_k^T \bm y^{(k)} = \bm r^{(k)}$
        \EndFor

        \vspace{1ex}
        \KwRet{$\bm Y = ((\bm y^{(0)})^T, \dots, (\bm y^{(K)})^T)^T$}
    \end{algorithmic}
\end{algorithm}

\section{Probabilistic scores}\label{prob_score}

A scoring rule $\text{S}$ assigns a real-valued quantity $\text{S}(F, \bm y)$ to a forecast-observation pair $(F, \bm y)$, where $F \in \mathcal{F}$ is a probabilistic forecast and $\bm y$ is an observation, possibly multivariate. Probabilistic scoring rules are used to assess the quality of predictive distributions, evaluating both calibration (how well predicted probabilities match observations) and sharpness (the concentration of the predictive distribution) of forecasts.

The Continuous Ranked Probability Score (CRPS)~\cite{gneiting2007} is the most popular univariate scoring rule ($y$ is univariate) and is defined as
$$\text{CRPS}(F,y)=\e_F\lvert X-y\rvert-\frac{1}{2}\e_F\lvert X-X'\rvert,$$
where $y \in \R$ and $X$ and $X'$ are independent random variables following $F$, with a finite first moment. The CRPS evaluates the accuracy of marginal predictive distributions at individual locations.

The Variogram Score (VS)~\cite{scheuerer2015} is a multivariate score designed to assess the dependence structure. The VS of order $p$ is defined as
$$\text{VS}_p(F, \bm y)=\sum_{i=1}^d\sum_{j=1}^d w_{i,j}\left(\e_F[\lvert y_i - y_j\rvert^p] - \e_F[\lvert X_i - X_j\rvert^p] \right)^2$$
where $X_i$ is the $i$-component of the random vector $X$ following $F$, $w_{i,j}$ are nonnegative weights that allows one to emphasize or down-weight pairs of component combinations based on subjective expert decisions, and $p>0$ is the order of the variogram score. The VS captures the quality of the forecasts’ spatial dependence by comparing predictions across location pairs.

\end{document}